\documentclass[11pt,a4paper]{article}

\pdftrailerid{}
\newif\ifanon
\anonfalse

\usepackage[margin=1in]{geometry}
\usepackage[T1]{fontenc}
\usepackage{lmodern}
\usepackage{microtype}
\usepackage{amsmath,amssymb,amsthm,mathtools}
\usepackage{array}
\usepackage{booktabs}
\usepackage[ruled,vlined]{algorithm2e}
\usepackage[font=small,labelfont=bf,tableposition=top]{caption}
\usepackage{placeins}
\usepackage{tikz}
\usetikzlibrary{arrows.meta,positioning}
\usepackage{xcolor}
\definecolor{linkblue}{rgb}{0,0.25,0.6}
\usepackage[colorlinks=true,linkcolor=linkblue,citecolor=linkblue,urlcolor=linkblue,bookmarksnumbered=true,pdfusetitle]{hyperref}
\usepackage{orcidlink}
\usepackage{aliascnt}
\usepackage{cleveref}

\newcommand{\sharedtheorem}[2]{%
  \newaliascnt{#1}{theorem}%
  \newtheorem{#1}[#1]{#2}%
  \aliascntresetthe{#1}}
\theoremstyle{plain}
\newtheorem{theorem}{Theorem}[section]
\sharedtheorem{lemma}{Lemma}
\sharedtheorem{proposition}{Proposition}
\sharedtheorem{corollary}{Corollary}

\numberwithin{equation}{section}

\crefname{theorem}{Theorem}{Theorems}
\Crefname{theorem}{Theorem}{Theorems}

\crefname{lemma}{Lemma}{Lemmas}
\Crefname{lemma}{Lemma}{Lemmas}

\crefname{proposition}{Proposition}{Propositions}
\Crefname{proposition}{Proposition}{Propositions}

\crefname{corollary}{Corollary}{Corollaries}
\Crefname{corollary}{Corollary}{Corollaries}

\crefname{algorithm}{Algorithm}{Algorithms}
\Crefname{algorithm}{Algorithm}{Algorithms}

\newcommand{\Z}{\mathbb{Z}}
\newcommand{\E}{\mathbb{E}}
\newcommand{\Prb}{\mathbb{P}}
\newcommand{\Fq}{\mathbb{F}_q}

\DeclareMathOperator{\Var}{Var}
\DeclareMathOperator{\Cov}{Cov}

\newcommand{\ignore}[1]{}

\newif\ifshowcomments
\showcommentsfalse
\ifshowcomments
\newcommand{\pnote}[1]{[{\footnotesize \textsf{\color{cyan}{\bf Prashant:} { {#1}}}}]}
\newcommand{\hnote}[1]{[{\footnotesize \textsf{\color{magenta}{\bf Haoxing:} { {#1}}}}]}
\newcommand{\cans}[1]{[{\footnotesize \textsf{\color{blue}{\bf Answer:} { {#1}}}}]}
\else
\newcommand{\pnote}[1]{}
\newcommand{\hnote}[1]{}
\newcommand{\cans}[1]{}
\fi

\title{%
  Provably Subexponential Algorithms for Bounded-Error LWE with Few Samples%
  \ifanon\else
    \thanks{%
      This work was supported by the National Research Foundation, Singapore, under its NRF Fellowship programme, award no.~NRF-NRFF14-2022-0010.%
    }%
  \fi
}
\ifanon
  \author{}
\else
  \author{%
    Haoxing Lin\,\orcidlink{0000-0001-9594-1871}
    and
    Prashant Nalini Vasudevan\,\orcidlink{0000-0001-6880-795X}\\[0.5ex]
    National University of Singapore\\
    \texttt{haoxingl@comp.nus.edu.sg}, \texttt{prashant@comp.nus.edu.sg}%
  }
  \hypersetup{pdfauthor={Haoxing Lin, Prashant Nalini Vasudevan}}
\fi
\date{}

\begin{document}

\maketitle

\begin{abstract}

\pnote{Please verify everything below, and fix if needed.} \cans{I have clarified the modulus regime, and prior-work comparison, and included the larger error range allowed for bounded secrets.}

We show that a variant of the Wagner/BKW-style algorithm for LWE of Kirchner and Fouque~\cite{KF15} solves decision and search LWE with secrets of size $n$ over prime moduli $q=n^{\kappa+o(1)}$, for any fixed $\kappa>0$, with errors bounded by $\exp((\ln n)^{o(1)})$, in time $2^{O(n/\ln\ln n)}$ given $n+O(n/\ln\ln n)$ original samples. Further, if each coordinate of the secret is also bounded by $\exp((\ln n)^{o(1)})$, the same time bound holds with errors bounded by $n^{o(1)}$ and only $O(n/\ln\ln q)$ samples.

Earlier, it was known that Kirchner and Fouque's algorithm could solve LWE with these parameters in subexponential time with a much larger number of samples $m = \Omega(n\log{n})$~\cite{HKM18}. Ducas et al.~\cite{DEL25} showed that an adaptation of this algorithm could solve the corresponding dual, the $\mathrm{SIS}^{\infty}$ problem, with a suitable norm bound $q/\operatorname{polylog}(n)$ and $m=n+\omega(n/\ln\ln n)$ variables in subexponential time. Whereas their analysis involves interpreting the algorithm in terms of chains of lattices, our approach is quite elementary and is based on bounding the lower moments of a Fourier statistic.

\end{abstract}

\newpage
\tableofcontents
\newpage

\section{Introduction}
\label{sec:paper-introduction}

Learning with errors (LWE) asks us to recover a secret from noisy linear equations, or to distinguish those equations from uniformly random data. Its connection to worst-case lattice problems and its use in public-key encryption make the resources needed to solve it an important cryptographic question~\cite{Reg05}. Time and memory are two such resources; the number of equations used is another. For example, Regev's public-key encryption scheme publishes a finite collection of noisy equations under one secret. For this reason, it is important to understand the complexity of the problem given a limited number of equations. We study this complexity for bounded errors and polynomial prime moduli. 

% A small error bound can help an algorithm, but the benefit depends on both the number of equations and the secret's range. A secret with polylogarithmically bounded coordinates has at most $2^{O(n\ln\ln n)}$ candidates, whereas an arbitrary field secret has $q^n=2^{\Theta(n\ln n)}$ candidates. This difference suggests that the two secret models may admit different requirements on the number of original samples. Counting candidates alone, however, gives neither an efficient recovery procedure nor a guarantee that combinations of the original equations behave independently. Our aim is to obtain a provable subexponential algorithm while keeping the number of original samples small in each model.

\subsection{Our results}
\label{sec:paper-results}

% \pnote{Can we really say that we give this algorithm? Because the algorithm is essentially from that paper from 2015 or so, right? The best we can say is that they gave the algorithm and made a heuristic claim (did they also give a wrong proof?), Ducal et al. gave a proof that it works for SIS, and we prove that it also works for LWE. Right? Whichever is the case, the theorem statement doesn't have to change, but we have to be very clear about the extent of our contributions and not claim too much.} \cans{I have credited the prior work here. Our construction differs from KF15 and DEL25 in its reuse of entries and independent subsampling of compatible pairs, as explained in Prior work.}

Building on the Wagner/BKW approach to LWE~\cite{KF15} and its provable SIS variant~\cite{DEL25}, we prove that a variant of the $k$-tree algorithm gives decision and search algorithms using a sublinear number of original samples for secrets bounded by $\exp((\ln n)^{o(1)})$ and errors bounded by $n^{o(1)}$, and only a sublinear excess over the secret dimension for arbitrary secrets and errors bounded by $\exp((\ln n)^{o(1)})$. With these sample counts, our algorithms succeed with high probability for every fixed secret-error pair satisfying the stated bounds (and not just random secrets and errors).

\paragraph{Model and theorem.} In LWE, the input consists of $m$ noisy linear equations $b_j=\langle a_j,s\rangle+e_j\bmod q$, with unknown secret $s\in\mathbb{F}_q^n$ and error $e_j \in \mathbb{F}_q$. We collect their public vectors $a_j$ as the columns of $A\in\Fq^{n\times m}$ and their right-hand sides in $b\in\Fq^m$, so $b=A^\top s+e\bmod q$, where $\Fq=\Z/q\Z$. The \emph{search} problem is to recover $s\bmod q$; the \emph{decision} problem is to distinguish such $(A,b)$ from the case where $b$ is uniform independently of $A$. We take $A$ to be uniformly random and fix the secret-error pair $(s,e)$ independently of $A$ (note that the error coordinates need not be independent). We denote by $E_s,E_e\ge0$ the bounds on the absolute values of each coordinate of the secret and error, respectively. % Time is counted in bit operations (bit time), and space in bits; a worst-case bound holds for every input and every outcome of the random choices. We use $\ln$ for the natural logarithm and $\log_2$ for the binary logarithm.

% Keep the hyperlink destination with the visible first line of the statement.
\begin{theorem}[LWE with few samples]
\leavevmode\phantomsection\label{thm:headline-lwe}Fix $\kappa>0$ and $C>\min\{1,1/\kappa\}$, and let $q=n^{\kappa+o(1)}$ run through primes. Consider the LWE problem with $m$ equations, secret of length $n$, and modulus $q$, with $E_s$ and $E_e$ being bounds on the absolute values of the secret and error coordinates, respectively. There is a decision (resp. search) algorithm for LWE in each of the following regimes that succeeds with advantage (resp. probability) $1-o(1)$ for all large enough $n$:
\begin{itemize}
\item $E_s\le\exp((\ln n)^{o(1)})$, $E_e\le n^{o(1)}$, and $m\ge Cn/\ln\ln q$ \pnote{Why is the bound on $E_e$ different from the bounds on $E_s$ and $E_e$ in the point below? Please check this carefully and make sure you understand the reason for the difference, as it looks suspicious} \cans{Verified. In the first regime, $E_e=n^{o(1)}$ still gives leaf width $Q=q^{1/2+o(1)}$, whereas the tighter secret bound preserves the running-time constant through the filter widths. In the second regime, normal form makes the original errors the new secret, so they must satisfy this tighter bound.}
\item $E_e\le\exp((\ln n)^{o(1)})$ and $m\ge n+Cn/\ln\ln q$
\end{itemize}
The expected running time of these algorithms is $2^{(1/2+o(1))n/\ln\ln q}$, and the same bound holds with probability at least $1-1/n$. They work for uniformly random $A$ and any (even worst-case) $s$ and $e$ fixed independently of $A$, as long as they satisfy the above bounds. % Here the $o(1)$ term and the minimum dimension for these guarantees may depend on $C$, $\kappa$, and the fixed asymptotic growth of the bounds. \pnote{Is there no dependence of the running time of the algorithm on the choice of $C$, $\kappa$, and $\beta$?} \cans{The $1/2$ constant is the same for every fixed $C,\kappa$ and fixed asymptotic growth of the bounds. Their dependence is in the $o(1)$ term and in how large $n$ must be. I have made this explicit and specified the expected and high-probability time guarantees.}
\end{theorem}

We prove \Cref{thm:headline-lwe} at the end of Section~\ref{sec:paper-reductions}.
% The fixed-pair guarantee also holds after averaging over any joint distribution of secret-error pairs satisfying these bounds and independent of $A$; error coordinates may be correlated. The resource bounds cover the entire search computation. \Cref{prop:reduction-resources} proves the runtime bound, termination with probability $1$, and worst-case space bound. The sample condition with fixed $C>1$ is sufficient for our analysis; its optimality remains open. 
The core is a decision algorithm for bounded secrets and errors. We obtain the search algorithm by adapting Regev's coordinate-guess reduction~\cite{Reg05}, and reduce arbitrary-secret instances to bounded-secret ones using the normal-form transformation of Applebaum--Cash--Peikert--Sahai~\cite{ACPS09}. The second regime has the tighter error bound because normal form makes the original errors the new secret, which must satisfy the first regime's secret bound. The algorithm is essentially a careful application of Wagner's $k$-tree algorithm~\cite{Wag02}. The formal decision algorithm is \Cref{alg:bounded-input-decision}, with its correctness stated in \Cref{prop:bounded-input-decision}; \Cref{sec:paper-reductions} gives the extensions to search and arbitrary secrets, and \Cref{prop:reduction-resources} bounds the resources of the full algorithms.

\paragraph{The $k$-tree algorithm.} Wagner's $k$-tree algorithm addresses the generalized birthday problem: choose one entry from each of $k$ lists of vectors so that their sum is zero~\cite{Wag02}. Its idea is to build candidate sums in stages. For four lists of binary vectors, with addition modulo two, first combine lists 1 and 2 and lists 3 and 4, keeping only sums that vanish on a chosen block of coordinates. Then combine the resulting two lists to find sums that vanish on the remaining coordinates as well. When $k$ is a power of two, these pairwise merges form a binary tree, with the input lists at its leaves and the final list at its root. Filtering limits intermediate list growth, while sums remain zero on coordinates processed at earlier levels.

For bounded-secret LWE decision, we use the first $d=\min\{m,n\}$ original equations and write $A,b,e$ for their restrictions. The goal is to find coefficient vectors $z$ that control error amplification and the secret's contribution through $Az$, yielding a detectable LWE signal. We draw $k$ initial lists at level zero, each with $N$ independently sampled coefficient vectors. Each vector specifies one weight for each equation. When combining two lists, adding entry vectors $x,y$ adds the corresponding weighted equations. At level $i$, the submatrix $A_i$ formed from $n_i$ rows of $A$ filters these sums, keeping only those for which every coordinate of $A_i(x+y)$ lies in the centered interval $\{-R_i,\ldots,R_i\}\bmod q$ (for a choice of $R_i$'s described later). We then independently subsample the passing sums, keeping each with probability $p_i$. Different levels use disjoint sets of rows, and the parameters are chosen to keep intermediate list sizes around the target size $N$.

\Cref{fig:paper-four-leaf-construction} illustrates the case $k=4$. First merge lists 1 and 2, and lists 3 and 4; then merge the two resulting lists. At each merge, every compatible pair with one entry from the left list and one from the right list is enumerated. Every entry remains available for every compatible pair, so one entry can contribute to many combinations. Each final entry is the sum of exactly one coefficient vector from each of the $k$ initial lists.

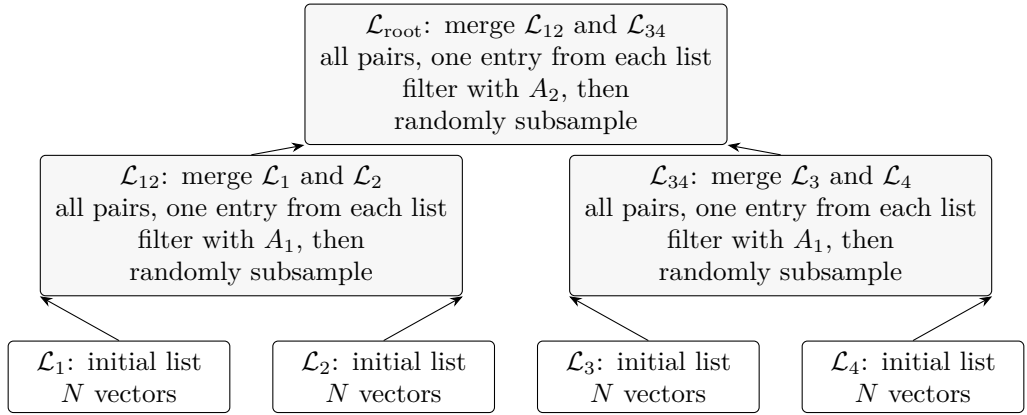
\begin{figure}[htbp]
\centering
\begin{tikzpicture}[x=1cm,y=1cm,>=Stealth,font=\small,leaf/.style={draw,rounded corners=2pt,align=center,text width=2.65cm,inner sep=4pt},merge/.style={draw,rounded corners=2pt,fill=black!3,align=center,text width=5.3cm,inner sep=4pt}]
  \foreach \j/\x in {1/-5.25,2/-1.75,3/1.75,4/5.25}
    \node[leaf] (leaf\j) at (\x,0) {$\mathcal L_{\j}$: initial list\\$N$ vectors};
  \node[merge] (left) at (-3.5,2.0) {$\mathcal L_{12}$: merge $\mathcal L_1$ and $\mathcal L_2$\\all pairs, one entry from each list\\filter with $A_1$, then randomly subsample};
  \node[merge] (right) at (3.5,2.0) {$\mathcal L_{34}$: merge $\mathcal L_3$ and $\mathcal L_4$\\all pairs, one entry from each list\\filter with $A_1$, then randomly subsample};
  \node[merge] (root) at (0,4.0) {$\mathcal L_{\mathrm{root}}$: merge $\mathcal L_{12}$ and $\mathcal L_{34}$\\all pairs, one entry from each list\\filter with $A_2$, then randomly subsample};
  \draw[->] (leaf1.north) -- (left.south west);
  \draw[->] (leaf2.north) -- (left.south east);
  \draw[->] (leaf3.north) -- (right.south west);
  \draw[->] (leaf4.north) -- (right.south east);
  \draw[->] (left.north) -- (root.south west);
  \draw[->] (right.north) -- (root.south east);
\end{tikzpicture}
\caption{The construction with four initial lists. Each merge adds one entry from each incoming list, filters the sum, and independently decides whether to keep it. Using an entry leaves it available for other pairs; equal vector values produced by different pairs remain separate entries.}
\label{fig:paper-four-leaf-construction}
\end{figure}

\FloatBarrier

\paragraph{The decision algorithm.} A coefficient vector $z\in\Fq^d$ combines the equations as:
\begin{equation}
\label{eq:intro-weight-identity}
 \langle z,b\rangle=\langle z,e\rangle+\langle Az,s\rangle\pmod q.
\end{equation}
The two terms are the combined error and the secret's contribution through the combined public vector. The initial coefficient interval controls error amplification, while the filters control the secret contribution. We map each combined right-hand side to a point on the complex unit circle and sum these points over the final list $\mathcal L_{\mathrm{root}}$:
\begin{equation}
\label{eq:intro-statistic}
 Z(b)=\sum_{z\in\mathcal L_{\mathrm{root}}}\exp(2\pi\mathrm{i}\langle z,b\rangle/q).
\end{equation}
Every list entry contributes separately, including duplicate vector values. The lists are generated using $A$ and fresh randomness, without using $b$. The construction ensures that all final coefficient vectors are nonzero, so when $b$ is uniform, every entry contributes a point with mean zero.
% Editorial note: The nonzero-vector invariant is stated above; its construction is explained in Section 3.
% Former explanatory sentence: Fixing the first coordinate of each initial vector to one, with $k<q$, ensures that all final vectors are nonzero.

For intuition, suppose the final list leads to the above $Z(b)$ being a sum of $N$ independent terms with a common mean of magnitude $N^{-o(1)}$. Their sum would have mean magnitude $N^{1-o(1)}$ and fluctuations of order $\sqrt N$. Under uniform data, its mean would be zero and its typical magnitude of order $\sqrt N$. A threshold $T$ near $N^{3/4}$ would separate these scales. The actual output terms are dependent; our analysis proves that this separation survives their dependence. \Cref{alg:tree-decision} summarizes the algorithm.

\begin{algorithm}[htbp]
\LinesNumbered
\caption{Bounded-secret LWE decision: informal outline}
\label{alg:tree-decision}
\KwIn{$d$ original equations, written as $A\in\Fq^{n\times d}$ and $b\in\Fq^d$; parameters, including list sizes, intervals, submatrices, and probabilities, chosen as in \Cref{tab:parameter-choices}.}
Draw $k=2^r<q$ initial lists of $N$ coefficient vectors $(1,X_2,\ldots,X_d)\bmod q$, where $X_2,\ldots,X_d$ are sampled independently and uniformly from $\{-B,\ldots,B\}$\;
\For{$i=1,\ldots,r$}{
  \ForEach{merge of two lists at level $i$}{
    Enumerate all pairs $(x,y)$ with $x$ from the left list, $y$ from the right list, and $A_i(x+y)\in[-R_i,R_i]^{n_i}\bmod q$\;
    Independently keep each sum $x+y$ in the parent list with probability $p_i$\;
    Leave both $x,y$ available for other pairs, and keep equal sums from different pairs as separate entries\;
  }
}
Compute $Z(b)$ from the final list by \eqref{eq:intro-statistic}\;
Output \textsc{LWE} if $|Z(b)|\ge T$, and \textsc{uniform} otherwise, with $T$ near $N^{3/4}$\;
\end{algorithm}

% \Cref{alg:tree-decision} summarizes the test. It specifies which pairs contribute; \Cref{sec:paper-resources} gives an implementation that finds compatible pairs efficiently.

\paragraph{Controlling dependence.} Independently sampled coefficient vectors still combine the same original equations, so the resulting terms need not be independent. We bound the Fourier statistic's mean and fluctuations under the actual dependent output distribution. We bound its second moment by tracking which merge decisions two outputs share when they reuse the same intermediate entry. We also show that accidental linear dependencies are rare enough to offset the greater chance of passing the filters. Sampling the remaining coordinates independently from a wide uniform interval makes these dependencies rare. The moment bounds in \Cref{thm:reusable-tree-moments} make the preceding separation rigorous without assuming independent outputs.

% \paragraph{Organization.} \Cref{sec:paper-reductions} derives the decision and search guarantees from the bounded-secret decision algorithm. \Cref{sec:paper-decision-tree} gives the formal construction and test. \Cref{sec:paper-moments} bounds dependence, \Cref{sec:paper-parameters} chooses the parameters and controls the signal, and \Cref{sec:paper-resources} gives the exact implementation and resource bounds.

\FloatBarrier
\subsection{Prior work}
\label{sec:paper-related-work}

\paragraph{LWE Hardness and the role of parameters.} Ajtai introduced the short integer solution problem (SIS), which asks for a short nonzero integer vector in the kernel of a matrix modulo $q$, and gave a worst-case to average-case reduction; Micciancio--Regev developed this connection using Gaussian measures~\cite{Ajtai96,MR07}. Regev introduced LWE and its quantum hardness reduction, and Brakerski--Langlois--Peikert--Regev--Stehl\'e established classical hardness for polynomial moduli~\cite{Reg05,BLP13}. For binary errors, Micciancio--Peikert prove hardness with suitably restricted samples and a change in lattice dimension~\cite{MP13}, whereas Arora--Ge's linearization method gives polynomial-time recovery from $m\ge c n^2$ samples, for a sufficiently large constant $c$, with independent uniform binary errors and polynomial prime moduli $q>2$~\cite{AG11}. Thus both the error bound and the number of equations matter.

\paragraph{BKW and algorithms with restricted number of equations.} BKW solves learning parity with noise (LPN), the binary-field analogue of LWE, with independent noise of any fixed rate below $1/2$ in time and samples $2^{O(n/\log n)}$~\cite{BKW03}. It combines equations to eliminate blocks of coordinates at the cost of increasing the noise. Lyubashevsky reduces the original samples to $n^{1+\epsilon}$, for any fixed $\epsilon>0$, at time $2^{O(n/\log\log n)}$ by generating combinations of the available equations~\cite{Lyu05}. These results illustrate a time--sample tradeoff over the binary field.

For LWE, lazy modulus switching combines equations whose public vectors agree after rounding to a smaller modulus, trading small residual coordinates for fewer buckets~\cite{AFFP14}. Kirchner--Fouque generalize this quantization and vary its precision between stages~\cite{KF15}; coded-BKW with sieving controls the accumulated error~\cite{GJMS17}. Further work improves practical implementations~\cite{BGJMSW20} and time--memory tradeoffs~\cite{WBLW25}. Guo, M\aa rtensson, and Stankovski Wagner experimentally study dependence from BKW-style combinations and sample amplification~\cite{GMSW21}. These experiments do not give an asymptotic guarantee for every fixed bounded secret-error pair, as we seek here.

Kirchner--Fouque claim subexponential time with a linear number of original samples and small secrets and errors~\cite{KF15}. For polylogarithmic secret and error bounds, polynomial moduli, and $m=\Theta(n)$, their claimed decision bound is $2^{(1/2+o(1))n/\ln\ln q}$~\cite[\href{https://arxiv.org/abs/1506.02717v4}{full version}, Cor.~6]{KF15}. As explained by Ducas--Engelberts--Loyer, Herold--Kirshanova--May identify a flaw in the linear-sample narrow-error proof and resolve it for a sufficiently large constant multiple of $n\log_2 n$ samples~\cite{HKM18,DEL25}. The gap concerns joint closeness to independent samples after amplification under a shared random matrix. For bounded secrets, we match the claimed running time with sublinear original samples. Our theorem proves correctness of our construction, not of the Kirchner--Fouque algorithm.

\paragraph{Wagner's algorithm for SIS and the obstacle for LWE.} Ducas--Engelberts--Loyer's Gaussian analysis of Wagner's algorithm~\cite{DEL25} builds on discrete-Gaussian techniques developed for shortest and closest vector problems~\cite{ADRS15,ADS15,AS18,ALS21,ACKS21}. They view BKW as Wagner's generalized birthday method~\cite{Wag02} on a dual lattice followed by the Aharonov--Regev distinguisher~\cite{AR05}. Their sampler lifts Gaussian vectors through a chain of projected lattices, with auxiliary superlattices keeping the lifts short. Gaussian randomized rounding lets them prove conditional similarity to independent samples, a notion from Aggarwal--Li--Stephens-Davidowitz~\cite{ALS21}. This gives a provable subexponential algorithm for $\mathrm{SIS}^{\infty}$ with $m=n+\omega(n/\ln\ln n)$ variables, polynomially growing prime moduli, and norm bounds $q/\operatorname{polylog}(n)$ for suitable polylogarithmic factors~\cite{DEL25}.

Transferring this SIS guarantee to LWE requires controlling a statistic that aggregates many vectors. Ducas--Engelberts--Loyer explain that bounding the joint statistical distance of all the required samples demands an accuracy for which their sampler takes exponential time in $n$~\cite[Section~4.3]{DEL25}.\footnote{Independent concurrent works of Han--Gao--Hu~\cite{HGH26} and Chen--Jiang--Wang--Zhang~\cite{CJWZ26} also extend the Gaussian--Wagner framework to LWE. We discuss their relation to our results and our different proof framework in Appendix~\ref{app:concurrent-work}.}

\paragraph{$k$-trees and dependence among combinations.} Levieil--Fouque's LPN variant forms all pairwise sums within each bucket, but treats its analysis as heuristic because the resulting equations are dependent~\cite[Section~4]{LF06}. For $k$-sum, Minder--Sinclair analyze extended $k$-trees with independent input lists, controlling dependence between candidate solutions; Lin--Vasudevan give success and complexity bounds over integers and modular groups~\cite{MinderSinclair12,LV24}. Joux--Kippen--Loss modify merging to enforce independence and uniformity of intermediate entries~\cite{JKL24}. Filtering coefficient vectors through a shared random matrix creates additional dependence beyond these analyses.

Kirchner--Fouque and Ducas--Engelberts--Loyer start from one list and discard both inputs after combining a pair~\cite{KF15,DEL25}. Our use of $k$ separate lists is a design choice. The coefficient vectors are sampled uniformly in their free coordinates from a wide interval. Every entry remains available for every compatible pair, and we independently subsample passing sums. Building on the second-moment approach to $k$-trees~\cite{MinderSinclair12,LV24}, we directly control the Fourier statistic's mean and fluctuations, including shared entries and linear dependencies under the common matrix. This gives the LWE guarantees in \Cref{thm:headline-lwe} without requiring joint closeness to independent samples.

\section{The full algorithm and its guarantees}
\label{sec:paper-decision-tree}

The core of our result is a decision algorithm for bounded secrets and errors. We first specify how this algorithm works. We then state the correctness and resource guarantees, followed by the two standard reductions that give search and handle arbitrary secrets. Sections~\ref{sec:paper-moments} and~\ref{sec:paper-resources} prove the correctness and resource statements conditioned on selecting good parameters; Section~\ref{sec:paper-parameters} chooses parameters that meet these conditions.

\paragraph{The LWE problems.} Let $q$ be prime. An instance consists of $A\in\Fq^{n\times m}$ and $b\in\Fq^m$. In the LWE case, $A$ is uniform and $b=A^\top s+e\bmod q$, where the secret $s\in\Z^n$ and error vector $e\in\Z^m$ are fixed independently of $A$. For bounded-secret LWE, we have the promise that $\|s\|_\infty\le E_s$ and $\|e\|_\infty\le E_e$, where $\|x\|_\infty=\max_j|x_j|$; in the arbitrary-secret case only the error bound is imposed. The search problem asks for $s\bmod q$. The decision problem asks us to distinguish an LWE instance from the uniform case, in which $b$ is uniform independently of $A$. 

In standard definitions of LWE, $s$ and $e$ are sampled from distributions respecting bounds similar to the ones above. Our algorithms work even for worst-case $s$ and $e$ as in the definition above, which is stronger than them working for fixed distributions. % Our decision guarantees bound both error probabilities for every fixed pair satisfying the promises; averaging also permits any joint distribution of such pairs independent of $A$.

\subsection{Constructing the lists}

The decision algorithm uses only the first $d=\min\{m,n\}$ original equations. Write $A\in\Fq^{n\times d}$ and $b\in\Fq^d$ for their matrix and right-hand-side vector. A coefficient vector $z\in\Fq^d$ specifies a linear combination of these equations. Our goal is to generate many coefficient vectors for which the combined error and the secret contribution are controlled, so that their Fourier terms have a detectable mean when $(A,b)$ is an LWE instance.

The parameter choices in \Cref{tab:parameter-choices} depend only on $n,q,d,E_s,E_e$. They specify a binary tree of depth $r$ with $2\le k=2^r<q$ initial lists, each of size $N=2^L$, a leaf radius $B$, and filter radii $R_1,\ldots,R_r$. They also specify a partition of the $n$ rows of $A$ into consecutive submatrices $A_i$ with positive row counts $n_i$ satisfying $\sum_i n_i=n$; write $s^{(i)}$ for the corresponding coordinates of the secret. We number levels upward from the initial lists at level zero. Write $M_i=k/2^i$ for the number of nodes at level $i$, each storing one list; for $i\ge1$, this is also the number of merges. We assume the parameter checks in Section~\ref{sec:parameter-computation} hold for the list count $k$, radii $B,R_i$, and row counts $n_i$; otherwise the algorithm returns \textsc{uniform}.

Sample every entry of the initial lists independently as
\begin{equation}
\label{eq:tree-leaf-distribution}
 (1,X_2,\ldots,X_d)\bmod q,\qquad
 X_2,\ldots,X_d\text{ independent and uniform on }\{-B,\ldots,B\}.
\end{equation}
All initial lists are sampled once. Fixing the first coordinate to one ensures that every vector at level $i$ has first coordinate $2^i\bmod q$. Since $k<q$, every vector generated by the tree is nonzero.

At level $i$, merge two child lists by adding their entry vectors. A pair $(x,y)$, with one entry from each child list, is \emph{compatible} when
\begin{equation}
\label{eq:tree-compatible-pair}
 A_i(x+y)\in\{-R_i,\ldots,R_i\}^{n_i}\bmod q.
\end{equation}
For each compatible pair, independently keep its sum in the parent list with probability $p_i$. Write $Q=2B+1$ and $W_i=2R_i+1$ for the interval widths, and set
\begin{equation}
\label{eq:tree-balancing}
 \rho_i=(W_i/q)^{n_i},\qquad
 p_i=\frac{q^{n_i}}{NW_i^{n_i}}=\frac1{N\rho_i}.
\end{equation}
The parameter checks and construction ensure $Q<q$, $3W_i\le q$, and $0<p_i\le1$. Before the fresh submatrix $A_i$ is revealed, a fixed nonzero pair sum has a uniform image, so it passes the filter and is kept with probability $\rho_i p_i=1/N$. Two lists of size $N$ would therefore produce $N$ parent entries in expectation. The same submatrix is used at every merge of a level; different levels use disjoint rows.

Using an entry leaves it available for other compatible pairs. Equal vector values remain separate entries, both in the initial lists and when different pairs produce the same sum. Each root entry is thus the sum of one selected entry from each of the $k$ initial lists. Figure~\ref{fig:paper-four-leaf-construction} illustrates the case $k=4$.

\paragraph{Finding compatible pairs with buckets.} Represent each coordinate of $A_i y$ by a residue in $\{0,\ldots,q-1\}$ and divide this range into consecutive intervals of width $W_i$, allowing a shorter final interval. We call each one-coordinate interval a \emph{cell}. A \emph{bucket} contains all right-list entries whose images fall in the same cell in every coordinate, so a bucket is indexed by a tuple of cell indices. Formally, the cells are
\begin{equation}
\label{eq:resource-cells}
 \mathcal I_{i,j}=\{jW_i,\ldots,\min\{(j+1)W_i-1,q-1\}\},\qquad
 0\le j<\lceil q/W_i\rceil.
\end{equation}
For a merge at level $i$, write the right list with size $J_{\mathrm R}$ as $\mathcal L_{\mathrm R}=(y_1, \dots, y_{J_{\mathrm R}})$. For a tuple $\mathbf j=(j_1,\ldots,j_{n_i})$ with $0\le j_t<\lceil q/W_i\rceil$, define the bucket
\begin{equation}
\label{eq:resource-buckets}
 \mathcal B_{i,\mathbf j}=\{a\in\{1,\ldots,J_{\mathrm R}\}:(A_i y_a)_t\in\mathcal I_{i,j_t}\text{ for every }t=1,\ldots,n_i\}.
\end{equation}
It stores the entries $y_a$ at these indices; equal vector values at different indices remain separate entries. \pnote{It is not clear what a cell or bucket is. Clearer explanation is needed here.}\cans{A cell is an interval of residues in one image coordinate; a bucket groups right-list entries lying in the same cell in every coordinate. I have explained these meanings before the formula.} \pnote{I am still having trouble understanding what a bucket is. Could you write it out mathematically, in the way you have written $\mathcal{I}{i,j}$ above?}\cans{I have defined each bucket explicitly by the indices of right-list entries whose images lie in the specified cells, retaining equal vectors at distinct indices.}

For each left entry $x$, a compatible right image must lie coordinatewise in $-(A_i x)_j+\{-R_i,\ldots,R_i\}\bmod q$. In each coordinate, find the cells meeting this cyclic interval, and query the buckets for all tuples of these cell indices. Check every returned pair against \eqref{eq:tree-compatible-pair} before sampling whether to keep it. Every compatible pair is found exactly once. The queries may also return pairs that fail the exact filter; they are discarded by the \eqref{eq:tree-compatible-pair} check before subsampling. Sorting by these index tuples allows binary search for each bucket; Section~\ref{sec:bucket-implementation} gives the implementation details and time bound.

\subsection{The decision algorithm}

Let $\omega=\exp(2\pi\mathrm{i}/q)$. For the analysis, let $\mathcal L_{\mathrm{root}}$ denote the root list obtained by running all merges without a storage limit, and define
\begin{equation}
\label{eq:tree-statistic}
 Z(b)=\sum_{z\in\mathcal L_{\mathrm{root}}}\omega^{\langle z,b\rangle}.
\end{equation}
Each entry contributes separately, including duplicate vector values. Generating the coefficient vectors uses only $A$ and fresh randomness; $b$ is used only when the sum $Z(b)$ is evaluated. In the uniform case, each nonzero vector's Fourier term has mean zero. Under LWE, $\langle z,b\rangle=\langle z,e\rangle+\langle Az,s\rangle\bmod q$, so the leaf radius $B$ controls the error contribution, and the filters control the secret contribution through the node images. We decide by the magnitude of the sum, using the threshold 
\begin{equation}
\label{eq:decision-threshold}
 T=2^{\lceil3L/4\rceil}.
\end{equation}

The ideal test returns \textsc{LWE} when $|Z(b)|\ge T$ and \textsc{uniform} otherwise. Here $|Z|$ is the complex modulus; the finite-precision implementation below applies the same rule to $\widetilde Z$ by comparing $|\widetilde Z|^2$ with $T^2$ exactly.

For the resource guarantees, we stop and return \textsc{uniform} if storing another entry would exceed $K$. We compute a rational estimate $\widetilde Z$ with $|\widetilde Z-Z|\le T/4$ whenever list construction finishes without exceeding this limit, and compare its squared modulus with $T^2$ exactly. These choices ensure bounded storage and reliable evaluation of the decision statistic; Section~\ref{sec:paper-resources} gives the implementation and its cost. Set
\begin{equation}
\label{eq:decision-integers}
 K=2kn^4N,\qquad
 P=\min\{j\in\Z_{\ge1}:2^jT\ge4K\}.
\end{equation}
Keep all stored entries until the algorithm ends, counting initial entries and entries retained by subsampling toward $K$; reaching exactly $K$ is permitted. All random sampling is exact, using fresh random bits for each request as described in Section~\ref{sec:resource-sampling}. On storage overflow, set $\widetilde Z=0$.

\begin{algorithm}[htbp]
\LinesNumbered
\caption{Bounded-secret LWE decision}
\label{alg:bounded-input-decision}
\KwIn{$A\in\Fq^{n\times m}$, $b\in\Fq^m$, and public bounds $E_s,E_e$.}
Use the first $d=\min\{m,n\}$ equations and compute the parameters in \Cref{tab:parameter-choices}; return \textsc{uniform} if a parameter check fails\;
Set the threshold $T$, storage limit $K$, and precision $P$ by \eqref{eq:decision-threshold} and \eqref{eq:decision-integers}\;
Draw $k$ initial lists of $N$ independent coefficient vectors from \eqref{eq:tree-leaf-distribution} using exact sampling (Section~\ref{sec:resource-sampling}); set the stored-entry count $J=kN$\;
\For{$i=1,\ldots,r$}{
  \ForEach{merge of two lists at level $i$, from left to right}{
    \ForEach{compatible pair $(x,y)$ found by bucket queries (Section~\ref{sec:bucket-implementation})}{
      Independently keep $x+y\bmod q$ with exact probability $p_i$\;
      \If{the sum is kept}{
        \lIf{$J=K$}{return \textsc{uniform}}
        Append the sum to the parent list and increment $J$\;
      }
    }
    Leave child entries available for other pairs; keep equal sums from different pairs as separate entries\;
  }
}
Compute the rational estimate $\widetilde Z$ with error at most $T/4$, using the evaluation procedure of Section~\ref{sec:resource-evaluation}\;
Return \textsc{LWE} if $|\widetilde Z|^2\ge T^2$, and \textsc{uniform} otherwise, comparing exactly\;
\end{algorithm}

\FloatBarrier
\subsection{Correctness}
\label{sec:decision-correctness}

Correctness rests on the mean and fluctuations of $Z$ under the actual reuse of entries. To describe the expected LWE signal, first consider one selection of a leaf entry from each initial list. At each internal node $v$ at level $i$, call $A_i z_v$ its \emph{node image}, where $z_v$ is the sum of the selected leaf vectors below $v$. For the root vector $z$, the combined right-hand side is $\langle z,e\rangle+\langle Az,s\rangle$. Its block $A_i z$ is the sum of the level-$i$ node images because those nodes partition the selected leaves.

To make the expected signal easier to calculate, we consider an \emph{ideal model} for one combination in which the node images are independent of one another and of the leaf vectors. In this model, sample $k$ independent leaf vectors from \eqref{eq:tree-leaf-distribution} and, at each internal node at level $i$, sample an independent uniform image in the filter box $\{-R_i,\ldots,R_i\}^{n_i}\bmod q$. Let $z^*\in\Fq^d$ be the sum of these leaf vectors, and let $y^*\in\Fq^n$ concatenate the sums of the ideal node images at levels $1,\ldots,r$. Thus $z^*$ combines the $d$ error coordinates, while the $i$th block of $y^*$ takes the place of $A_i z$ and pairs with the corresponding secret coordinates. For a fixed integer secret-error pair $(s,e)\in\Z^n\times\Z^d$, the \emph{ideal mean} below is the mean Fourier contribution in this independent model.
\begin{equation}
\label{eq:tree-ideal-mean}
 \Gamma=\E\bigl[\omega^{\langle z^*,e\rangle+\langle y^*,s\rangle}\bigr].
\end{equation}
The algorithm does not compute $\Gamma$, which depends on the unknown secret and errors. We use it in the analysis: the following theorem shows that the actual mean of the sum is close to $N\Gamma$ and bounds its fluctuations despite the dependence among outputs.

\begin{theorem}[Moments of the Fourier statistic]
\label{thm:reusable-tree-moments}
Let $A\in\Fq^{n\times d}$ be uniform and independent of the initial lists and subsampling randomness. Suppose the algorithm parameters additionally satisfy
\begin{equation}
\label{eq:moment-finite-conditions}
 \begin{gathered}
 2k\le N,\qquad D_0=d-1-2k\ge1,\qquad 4kN^{r+1}\le Q^{D_0},\\
 \varepsilon=4kN^rQ^{-D_0}\le N^{-1}.
 \end{gathered}
\end{equation}
For every fixed integer secret-error pair, with $b=A^\top s+e\bmod q$,
\begin{equation}
\label{eq:reusable-moments}
 |\E[Z(b)]-N\Gamma|\le3N\varepsilon\le3,\qquad \Var[Z(b)]\le16kN,
\end{equation}
where $\Var[Z]=\E[|Z-\E[Z]|^2]$. If $b$ is uniform independently of the matrix and tree randomness, then
\begin{equation}
\label{eq:reusable-null-moment}
 \E_{\mathrm{unif}}[|Z(b)|^2]\le10kN.
\end{equation}
For every tree node $v$, its list $\mathcal L_v$ in the construction without a storage limit satisfies $\E[|\mathcal L_v|]\le(1+\varepsilon)N$.
\end{theorem}

Section~\ref{sec:moment-comparison} proves \cref{thm:reusable-tree-moments} by accounting for shared intermediate entries and accidental linear dependencies among leaf vectors. With $k=N^{o(1)}$ and $|\Gamma|=N^{-o(1)}$, the LWE mean has magnitude $N^{1-o(1)}$, while the fluctuations in both cases have scale at most $N^{1/2+o(1)}$. The threshold near $N^{3/4}$ separates these scales. Taking a magnitude also removes the unknown rotation $\omega^{ke_1}$ contributed by the fixed first coordinates.

\begin{lemma}[Decision from an accurate estimate]
\label{lem:decision-accurate-estimate}
Suppose the moment bounds of \Cref{thm:reusable-tree-moments} hold and $3+5T/4\le N|\Gamma|/2$. Suppose the computed estimate satisfies $\Prb[|\widetilde Z-Z|>T/4]\le\eta$ both under LWE with the fixed pair $(s,e)$ and under uniform input, so $\eta$ bounds the probability that the estimate differs from $Z$ by more than $T/4$. The test returns \textup{\textsc{LWE}} when $|\widetilde Z|\ge T$ and \textup{\textsc{uniform}} otherwise. The subscripts below indicate the input model; probabilities include the uniform matrix, the algorithm's random choices, and, in the uniform case, the independent uniform vector $b$. Its probability of returning \textup{\textsc{uniform}} on an LWE input with the fixed pair $(s,e)$, and of returning \textup{\textsc{LWE}} on a uniform input, respectively satisfy
\begin{equation}
\label{eq:decision-alternative-error}
 \Prb_{\mathrm{LWE}}[|\widetilde Z|<T]\le\eta+\frac{64k}{N|\Gamma|^2},
\end{equation}
\begin{equation}
\label{eq:decision-null-error}
 \Prb_{\mathrm{unif}}[|\widetilde Z|\ge T]\le\eta+\frac{160kN}{9T^2}.
\end{equation}
\end{lemma}

We prove \Cref{lem:decision-accurate-estimate} in Section~\ref{sec:decision-moment-application}.

\Cref{thm:bit-resources} bounds the probability of stopping because the stored-entry count would exceed $K$; \Cref{cor:decision-accuracy} therefore supplies the estimate condition with $\eta=2n^{-4}$. \Cref{lem:schedule-properties} shows that parameters can be set to satisfy $|\Gamma|=N^{-o(1)}$ and $k=N^{o(1)}$, which imply the threshold separation in \Cref{lem:decision-accurate-estimate} for sufficiently large $n$. Together these results give the following decision guarantee.

\begin{proposition}[Bounded-secret decision]
\label{prop:bounded-input-decision}
Fix $\kappa>0$ and $C>\min\{1,1/\kappa\}$. Let $q=n^{\kappa+o(1)}$ be prime, let $m\ge Cn/\ln\ln q$ be polynomial in $n$, and let public rational bounds $0\le E_s\le\exp((\ln n)^{o(1)})$ and $0\le E_e\le n^{o(1)}$ have polynomial encoding length. For every fixed integer pair $(s,e)\in\Z^n\times\Z^m$ with $\|s\|_\infty\le E_s$ and $\|e\|_\infty\le E_e$, independent of uniform $A$, \Cref{alg:bounded-input-decision} has error at most $3n^{-4}$ for sufficiently large $n$. The same bound holds for uniformly random instances and for each fixed asymptotic growth of the public bounds throughout the stated parameter ranges.
\end{proposition}

We prove \Cref{prop:bounded-input-decision} in Section~\ref{sec:decision-moment-application}.

\subsection{Time and space complexity}

We measure time in bit operations and space in bits, including temporary storage. Thus the cost of arithmetic depends on the bit lengths of its operands.

The bucket merge, storage cap, and numerical evaluation turn the statistical test into an algorithm with bounded resources. The next theorem states their guarantees under explicit conditions on the parameters; its proof is in Section~\ref{sec:paper-resources}.

\begin{theorem}[Total resources and numerical accuracy]
\label{thm:bit-resources}
Suppose the algorithm parameters are computable in polynomial time, have polynomial encoding length, and satisfy \eqref{eq:moment-finite-conditions}. Assume $q$, $m$, and the encoding lengths of the public bounds are polynomial in $n$ and, at every level,
\begin{equation}
\label{eq:resource-conditions}
 \log n=o(L),\qquad 3^{n_i}=N^{o(1)},\qquad N\rho_i=N^{o(1)}.
\end{equation}
Then \Cref{alg:bounded-input-decision} terminates with probability $1$ and uses worst-case space $N^{1+o(1)}$ for every input, including temporary storage and all random sampling trials. For uniform $A$ independent of the tree randomness, its expected running time is $N^{1+o(1)}$ and the probability of exceeding the storage limit is at most $\tfrac32n^{-4}$. When list construction finishes without exceeding the storage limit, evaluation and the final comparison use at most $K\operatorname{poly}(n)=N^{1+o(1)}$ time for every right-hand-side vector $b$.

When list construction finishes without exceeding the storage limit, $|\widetilde Z-Z(b)|\le K2^{-P}\le T/4$, and the squared-modulus comparison is exact, with equality returning \textup{\textsc{LWE}}. These guarantees hold throughout any parameter range satisfying the hypotheses uniformly.
\end{theorem}

\begin{corollary}[Accuracy of the decision statistic]
\label{cor:decision-accuracy}
Under the hypotheses of \Cref{thm:bit-resources}, for uniform $A$ independent of the tree randomness, the computed statistic satisfies
\begin{equation}
\label{eq:resource-accuracy-failure}
 \Prb[|\widetilde Z-Z|>T/4]\le2n^{-4},
\end{equation}
both on LWE inputs with every fixed secret-error pair and on uniform inputs.
\end{corollary}

\begin{proof}
Outside storage overflow, the numerical error is at most $T/4$. The probability of overflow is at most $\tfrac32n^{-4}\le2n^{-4}$ by \Cref{thm:bit-resources}, and is the same in both input models because list construction does not use $b$.
\end{proof}

\begin{lemma}[Parameter choice]
\label{lem:schedule-properties}
Fix $\kappa>0$ and $C>\min\{1,1/\kappa\}$, and let $q=n^{\kappa+o(1)}$ run through primes. For each fixed asymptotic growth of the public bounds and all sufficiently large $n$, the parameters chosen in \Cref{tab:parameter-choices} are computable in polynomial bit time for every integer dimension $d$ and public rational bounds of polynomial encoding length satisfying
\begin{equation}
\label{eq:schedule-uniform-range}
 Cn/\ln\ln q\le d\le n,\qquad 0\le E_s\le\exp((\ln n)^{o(1)}),\qquad 0\le E_e\le n^{o(1)}.
\end{equation}
All parameter checks pass, \eqref{eq:moment-finite-conditions} and \eqref{eq:resource-conditions} hold, and
\begin{equation}
\label{eq:schedule-exponent}
 L=\left(\tfrac12+o(1)\right)\frac n{\ln\ln q},\qquad k=N^{o(1)},\qquad |\Gamma|=N^{-o(1)}.
\end{equation}

The estimate $|\Gamma|=N^{-o(1)}$ holds for every fixed secret-error pair $(s,e)$ with $\|s\|_\infty\le E_s$ and $\|e\|_\infty\le E_e$. For each fixed asymptotic growth of the bounds, the estimates are uniform in $d$ and the secret-error pair, including $q<n$ and zero bounds.
\end{lemma}

Section~\ref{sec:parameter-verification} proves \cref{lem:schedule-properties}. Thus one invocation of the bounded-secret decision algorithm has expected running time and worst-case space $2^{(1/2+o(1))n/\ln\ln q}$. The reductions below make polynomially many such invocations and execute them sequentially, so their time costs add in expectation and their workspace can be reused.

\begin{proposition}[Resources across all calls]
\label{prop:reduction-resources}
Under the hypotheses of \Cref{thm:headline-lwe}, the decision and search algorithms of \Cref{sec:paper-reductions} have expected running time and worst-case space at most $2^{(1/2+o(1))n/\ln\ln q}$. With probability at least $1-1/n$, their running time is also at most $2^{(1/2+o(1))n/\ln\ln q}$. They terminate with probability $1$. These bounds include selecting and transforming the input equations, computing the parameters, all invocations of the bounded-secret decision algorithm, and all random sampling, including rejected draws. The high-probability time bound is a bound on runs taken to completion, so it introduces no additional correctness failure.
\end{proposition}

We prove \Cref{prop:reduction-resources} in Section~\ref{sec:reduction-resources}.

\subsection{Search and arbitrary secrets}
\label{sec:paper-reductions}

Two standard transformations extend the bounded-secret decision algorithm to all the cases in the main theorem. Both use only the original input samples; search reuses these samples across all decision calls. Appendix~\ref{app:reduction-proofs} verifies the adaptations for fixed secret-error pairs and the explicit failure bounds.

\paragraph{Recovering a bounded secret.} Regev's coordinate-guess transformation~\cite[\href{https://cims.nyu.edu/~regev/papers/qcrypto.pdf}{full version}, Lemma~4.2]{Reg05} tests a possible value $\sigma$ of secret coordinate $j$ by adding a fresh uniform $w\in\Fq^d$ to row $j$ of $A$ and adding $\sigma w$ to $b$. A correct guess preserves the LWE instance's secret and errors; an incorrect guess makes the right-hand-side vector uniform given the perturbed matrix. Test each guess on a fresh perturbation of the same first $d$ samples, using fresh decision randomness. Return the unique accepted value at each coordinate, reporting failure if it is missing or not unique. 

\begin{proposition}[Bounded-secret search]
\label{prop:bounded-search}
Under the hypotheses of \Cref{prop:bounded-input-decision}, let $\mathcal S=\{-\lfloor E_s\rfloor,\ldots,\lfloor E_s\rfloor\}$ be the set of candidate values for each secret coordinate. Search tests every $\sigma\in\mathcal S$ for every coordinate $j=1,\ldots,n$ using the perturbation described above, the same first $d=\min\{m,n\}$ original samples, and fresh randomness for each test. It returns the unique accepted value at each coordinate, reporting failure if there is none or more than one. The number of scheduled decision calls is
\begin{equation}
\label{eq:reduction-guess-count}
 M_{\mathrm{calls}}=n|\mathcal S|=n(2\lfloor E_s\rfloor+1).
\end{equation}
For every fixed secret-error pair satisfying the bounds, failure or incorrect recovery has probability at most $3M_{\mathrm{calls}}n^{-4}=n^{-3+o(1)}$ for sufficiently large $n$.
\end{proposition}

We prove \Cref{prop:bounded-search} in Appendix~\ref{app:bounded-search}.

\paragraph{Replacing an arbitrary secret by errors.} The normal-form transformation of Applebaum--Cash--Peikert--Sahai~\cite[Section~3.1, Lemma~2]{ACPS09} uses an invertible block of original equations to make their errors the new secret, up to sign. We reserve the final $n$ samples and test their matrix for invertibility once. On singularity, decision returns \textsc{uniform} and search reports failure. Otherwise, apply normal form and use the bounded-secret algorithms with $E_s=E_e=E$.

\begin{proposition}[Normal form using reserved samples]
\label{prop:fixed-suffix-normal-form}
Let $q$ be prime and $m>n$. There is a reduction, computable in time polynomial in the input length, from LWE with $n$ secret coordinates, $m$ equations, an arbitrary secret, and errors bounded by $E$ to LWE with $n$ secret coordinates, $m-n$ equations, and both secret and errors bounded by $E$. It uses only the original samples and fails only when the matrix of the final $n$ equations is singular, which happens with probability at most $1/(q-1)$. For every fixed original secret-error pair independent of the uniform input matrix, conditional on any fixed invertible reserved matrix, the output matrix is uniform and the transformed secret-error pair is fixed. Uniform input remains uniform. A recovered transformed secret determines the original secret using the reserved equations.
\end{proposition}

Appendix~\ref{app:fixed-suffix-normal-form} proves \cref{prop:fixed-suffix-normal-form} and gives the transformation and recovery formulas. The bounded-secret algorithms use only the first $d=\min\{m-n,n\}$ transformed equations. Transforming just those equations is equivalent to transforming all $m-n$ and then discarding the unused ones.

\begin{corollary}[Arbitrary-secret decision and search]
\label{cor:arbitrary-secret}
Fix $\kappa>0$ and $C>\min\{1,1/\kappa\}$. Let $q=n^{\kappa+o(1)}$ be prime, let $m\ge n+Cn/\ln\ln q$ be polynomial in $n$, and let the public rational bound $0\le E\le\exp((\ln n)^{o(1)})$ have polynomial encoding length. For sufficiently large $n$, for every fixed secret $s\in\Fq^n$ and fixed integer error vector $e\in\Z^m$ with $\|e\|_\infty\le E$, independent of uniform $A$, normal form followed by bounded-secret decision or search with $E_s=E_e=E$ has respective error bounds
\begin{equation}
\label{eq:reduction-arbitrary-errors}
 3n^{-4}+\frac1{q-1},\qquad
 3n(2\lfloor E\rfloor+1)n^{-4}+\frac1{q-1}.
\end{equation}
The same decision error bound holds for uniform input. Both algorithms use only the original samples; search keeps the reserved equations unchanged throughout its decision tests and recovery.
\end{corollary}

We prove \Cref{cor:arbitrary-secret} at the end of Appendix~\ref{app:fixed-suffix-normal-form}.

\begin{proof}[Proof of \Cref{thm:headline-lwe}]
For bounded secrets and errors, apply \Cref{prop:bounded-input-decision,prop:bounded-search}; for arbitrary secrets, apply \Cref{cor:arbitrary-secret}. All displayed error bounds tend to zero. \Cref{prop:reduction-resources} gives the total time, space, and termination guarantees. These bounds hold for every matrix-independent fixed secret-error pair, so averaging also allows any joint distribution satisfying the bounds and independent of the matrix.
\end{proof}

\section{Correctness under sample reuse}
\label{sec:paper-moments}

Algorithm~\ref{alg:bounded-input-decision} constructs coefficient-vector lists using $A$ and fresh randomness. For the analysis, let $\mathcal L_{\mathrm{root}}$ be its root list without a storage limit and write $Z(b)=\sum_{z\in\mathcal L_{\mathrm{root}}}\omega^{\langle z,b\rangle}$. Under LWE, we expect this sum to have mean close to $N\Gamma$, where $\Gamma$ is the mean of one Fourier contribution in the ideal model with independent leaf vectors and independent uniform node images in their filter boxes. Correctness requires the fluctuations, and the statistic's size under uniform input, to be small enough for the decision threshold to separate the two cases. \pnote{There is not enough context here for the reader to understand this paragraph, or even the issue that it is addressing. Another pargraph is needed that sets up the context, and also points the reader to the algorithm.}\cans{Added context recalling the algorithm, Fourier statistic, ideal mean, and the separation needed for decision before explaining the two sources of dependence.}

\paragraph{Moment bounds (Theorem~\ref{thm:reusable-tree-moments}).} For uniform $A$ independent of the tree randomness, under the theorem's parameter hypotheses, including \eqref{eq:moment-finite-conditions}, the construction without a storage limit satisfies \pnote{Restate the relevant theorem here}\cans{Restated the moment bounds here, including the LWE mean and variance, uniform-case second moment, and expected list sizes.}
\begin{equation*}
\begin{gathered}
 |\E_{\mathrm{LWE}}[Z(b)]-N\Gamma|\le3N\varepsilon\le3,\qquad \Var_{\mathrm{LWE}}[Z(b)]\le16kN,\\
 \E_{\mathrm{unif}}[|Z(b)|^2]\le10kN,\qquad \E[|\mathcal L_v|]\le(1+\varepsilon)N\quad\text{for every node }v.
\end{gathered}
\end{equation*}
The LWE bounds hold for every secret-error pair fixed independently of the uniform matrix; the uniform bound holds when $b$ is uniform independently of the matrix and tree randomness.

The terms of $Z(b)$ can be dependent for two reasons. Two leaf selections may use the same leaf entries and hence share intermediate entries and entry randomness. Even leaf selections with no shared leaf entries use the same matrix, and a linear dependency among their leaf vectors can make several filters easier to pass. We account for shared entries by counting shared subtrees, and for the common matrix by bounding both the rarity of linear dependencies among leaf vectors and the advantage they give to survival.

\subsection{Describing entries by leaf selections}
\label{sec:paper-leaf-selections}

For the analysis, describe a possible root entry by choosing one leaf entry from each leaf list. Denote the $a$th vector in list $j$ by $X_{j,a}$, and write $a=(a_1,\ldots,a_k)\in[N]^k$ for the \emph{leaf selection}, a tuple of entry indices, where $[N]=\{1,\ldots,N\}$. These choices determine every intermediate sum, whether or not the entry is produced.

Denote the root node by $\mathrm{root}$. For a node $v$, let $\operatorname{ht}(v)$ be its level and $I(v)\subseteq[k]$ the interval of leaf positions below it; in particular, $I(\mathrm{root})=[k]$. Put\pnote{Instead of the notation $z_a$, it would be better to give a special name for the root node and then just use $z_{root name}(a)$, consistent with the other $z_v(a)$}\cans{Named the root $\mathrm{root}$ and consistently replaced $z_a$ by $z_{\mathrm{root}}(a)$.}
\begin{equation}
\label{eq:tree-witness}
 X_a=(X_{1,a_1},\ldots,X_{k,a_k}),\qquad
 z_v(a)=\sum_{j\in I(v)}X_{j,a_j}.
\end{equation}
The possible entry at $v$ corresponding to $a$ is identified by $(v,a|_{I(v)})$. Restricting $a$ to the two child subtrees determines its unique producing pair. Induction up the tree gives a one-to-one correspondence between possible entries at $v$ and restricted leaf selections, even when different selections have equal vector sums. Every leaf vector has first coordinate one, so the first-coordinate invariant is
\begin{equation}
\label{eq:tree-anchor}
 (z_v(a))_1=2^{\operatorname{ht}(v)}\bmod q\ne0.
\end{equation}
The nonzero value follows from $2^{\operatorname{ht}(v)}\le k<q$.

Represent the entry randomness for every possible internal entry by a Bernoulli subsampling variable $U_{v,a|_{I(v)}}$, equal to one with probability $p_{\operatorname{ht}(v)}$. These variables are mutually independent and independent of the leaf vectors and $A$. An entry's variable is only used by the construction when its two children exist and their pair passes the filter. This lets us describe survival over a fixed family of leaf selections:
\begin{equation}
\label{eq:tree-survival}
 C_a=\prod_{v\ \mathrm{internal}}U_{v,a|_{I(v)}}\,\mathbf1_{\{A_{\operatorname{ht}(v)}z_v(a)\in[-R_{\operatorname{ht}(v)},R_{\operatorname{ht}(v)}]^{n_{\operatorname{ht}(v)}}\bmod q\}}.
\end{equation}
Thus $C_a=1$ precisely when all merges required by the leaf selection succeed. For any node $v$, let $C_{v,a}$ denote the product in \eqref{eq:tree-survival} restricted to the internal nodes in the subtree rooted at $v$, including $v$ when it is internal. Then $C_{v,a}=1$ precisely when the restricted leaf selection $a|_{I(v)}$ produces an entry at $v$; for a leaf, the empty product is one. The Fourier statistic can now be written as
\begin{equation}
\label{eq:tree-witness-statistic}
 Z(b)=\sum_{a\in[N]^k}C_a\omega^{\langle z_{\mathrm{root}}(a),b\rangle}.
\end{equation}

To control the fluctuations of $Z(b)$, we expand its second moment as a sum over ordered pairs of leaf selections:
\begin{equation}
\label{eq:tree-second-moment-pairs}
 \E[|Z(b)|^2]=\sum_{a,a'\in[N]^k}\E\bigl[C_aC_{a'}\omega^{\langle z_{\mathrm{root}}(a)-z_{\mathrm{root}}(a'),b\rangle}\bigr].
\end{equation}
Each Fourier term has magnitude one, so the absolute value of a pair's expected contribution is at most $\E[C_aC_{a'}]$, the probability that both selections survive. We therefore need to understand which merge decisions a pair shares.

Two leaf selections agreeing on every leaf entry below an internal node reuse the same intermediate entry, including its filter test and entry randomness. For example, in the four-leaf tree of \Cref{fig:paper-four-leaf-construction}, two leaf selections agreeing exactly in lists 1 and 2 share the first left merge. If they instead agree exactly in lists 1 and 3, they share two leaf entries but no internal entry. The arrangement of shared leaf entries determines which merge decisions are reused.

For the moment, suppose the distinct leaf vectors used by a pair are linearly independent. Its distinct merge decisions then succeed independently with probability $\rho_ip_i=1/N$ each, by \eqref{eq:tree-balancing}. With no shared internal entry, the joint survival probability is $N^{-2(k-1)}$. Since there are $N^{2k}$ ordered pairs of leaf selections in total, we use $N^{2k}N^{-2(k-1)}=N^2$ as the reference scale for the sum of pair-survival probabilities.

Of these $N^{2k}$ pairs, exactly $N^{2k-2}$ agree at two specified leaf positions; requiring no agreement elsewhere can only reduce this count. Thus each of the two sharing patterns above has a counting factor at most $N^{-2}$.

Sharing the entries in lists 1 and 2 also reuses one complete merge in their shared subtree: its filter and subsampling decision are required once rather than twice. Replacing two factors $1/N$ by one increases the joint survival probability by a factor $N$. Sharing the entries in lists 1 and 3 reuses no merge. Combining the counting and survival factors gives base weights $N^{-1}$ for sharing lists 1 and 2 and $N^{-2}$ for sharing lists 1 and 3. These weights bound the corresponding pair contributions after division by the reference scale $N^2$, before accounting for linear dependencies among leaf vectors. \Cref{lem:shared-forest} extends this count to arbitrary shared subtrees and includes the effect of those dependencies.

\begin{table}[htbp]
\centering
\caption{Notation for the tree and its moment analysis.}
\label{tab:paper-notation}
\begin{tabular}{@{}>{\raggedright\arraybackslash}p{.29\textwidth}>{\raggedright\arraybackslash}p{.67\textwidth}@{}}
\toprule
\textbf{Symbols} & \textbf{Meaning}\\
\midrule
$n,m,d,q$ & Secret dimension, number of original samples, number of selected samples, prime modulus.\\
$A,b,s,e$; $E_s,E_e$ & Matrix, right-hand-side vector, secret, errors; public coordinate bounds.\\
$k=2^r$, $N=2^L$ & Number of leaf lists and tree depth; leaf list size and its exponent.\\
$A_i,n_i,M_i$; $s^{(i)}$ & Level-$i$ submatrix, row count, node count; corresponding secret block.\\
$B,Q$; $R_i,W_i$ & Leaf and filter radii, with widths $Q=2B+1$, $W_i=2R_i+1$.\\
$\rho_i,p_i$ & Probability that a uniform image passes the filter; probability with which a passing sum is kept, with $\rho_ip_i=1/N$.\\
$U_{v,a|_{I(v)}}$ & Bernoulli subsampling variable for a possible internal entry at $v$.\\
$\mathcal L_v$, $a$, $C_a,C_{v,a}$ & List at node $v$; leaf selection (a tuple of entry indices); survival indicators at the root and at $v$.\\
$z_v(a),z_{\mathrm{root}}(a)$ & A leaf selection's sum at node $v$ and at the root.\\
$Z(b),\Gamma,T$ & Fourier statistic, ideal mean, and decision threshold.\\
\bottomrule
\end{tabular}
\end{table}

\FloatBarrier
\subsection{The moment bounds}

Use the leaf selections and survival indicators just defined. Enumerate the internal nodes as $v_1,\ldots,v_{k-1}$ by increasing level and then from left to right. For a leaf selection $a$, define the ordered tuple of its node images by
\begin{equation}
\label{eq:ordered-node-images}
 Y_a=\bigl(A_{\operatorname{ht}(v_1)}z_{v_1}(a),\ldots,A_{\operatorname{ht}(v_{k-1})}z_{v_{k-1}}(a)\bigr).
\end{equation}

We prove the mean and variance bounds more generally for any fixed complex-valued contribution function $g(X,Y)$ satisfying $|g(X,Y)|\le1$ for all inputs. For fixed integer secret and errors $(s,e)\in\Z^n\times\Z^d$, independent of $A$, the LWE contribution is obtained with \pnote{Expand this out here and point out what the function $g$ is in this case. Otherwise it is a bit confusing.}\cans{Moved the explicit LWE contribution function to this first use, explained its relation to the node images, and removed the later repetition.}
\begin{equation}
\label{eq:fixed-lwe-character-function}
 g((x_j)_j,(y_v)_v)=\omega^{\textstyle\left\langle\sum_jx_j,e\right\rangle+\sum_{i=1}^r\left\langle\sum_{\operatorname{ht}(v)=i}y_v,s^{(i)}\right\rangle}.
\end{equation}
This function has absolute value one. Under LWE, substituting $b=A^\top s+e$ gives $\langle z_{\mathrm{root}}(a),b\rangle=\langle z_{\mathrm{root}}(a),e\rangle+\langle Az_{\mathrm{root}}(a),s\rangle\bmod q$. The level-$i$ nodes partition the selected leaves, so their actual images sum to $A_i z_{\mathrm{root}}(a)$. Hence $g(X_a,Y_a)=\omega^{\langle z_{\mathrm{root}}(a),b\rangle}$. A surviving leaf selection contributes $g(X_a,Y_a)$, and a failed leaf selection contributes zero. Put
\begin{equation}
\label{eq:moment-observable}
 \xi_a=C_ag(X_a,Y_a),\qquad Z_g=\sum_{a\in[N]^k}\xi_a,\qquad \Gamma_g=\E[g(X^*,Y^*)].
\end{equation}
Here $(X^*,Y^*)$ is sampled from the ideal model for one leaf selection defined in Section~\ref{sec:decision-correctness}. Recall that its leaf vectors are sampled independently from \eqref{eq:tree-leaf-distribution}, and its node images are independent uniform points in their respective filter boxes, independent of the leaf vectors. Thus, for this LWE contribution function, $\Gamma_g=\Gamma$ from \eqref{eq:tree-ideal-mean}.

The function $g$ is the same for every leaf selection and depends only on its selected leaf vectors and node images; any additional parameters are fixed independently of the matrix, leaf vectors, and entry randomness. Thus every leaf selection has the same ideal mean $\Gamma_g$. We prove the mean and variance conclusions of \Cref{thm:reusable-tree-moments} with $Z_g,\Gamma_g$ in place of $Z,\Gamma$.\pnote{Refer to the theorem, and not just the specific equation within the theorem}\cans{Referred to the theorem and specified that the generalization concerns its mean and variance conclusions.}

The proof has three steps. First we bound linear dependencies among leaf vectors, weighted by their advantage in passing the filters. Then we count pairs of leaf selections sharing subtrees. Finally we combine these bounds to control the mean, covariance, and uniform-case second moment. All sums count entries separately, including equal vector values.

\subsection{The survival advantage of a linear dependency among leaf vectors}
\label{sec:leaf-rank-survival}

The wide intervals of the leaf distribution make linear dependencies rare, but a dependency can increase the chance of passing several filters. Proposition~\ref{prop:affine-concentration} bounds the probability that a sampled leaf vector lies in the affine subspace determined by earlier vectors and its fixed first coordinate. Proposition~\ref{prop:rank-filtering} bounds filter probabilities in terms of rank, while Proposition~\ref{prop:independent-filtered-images} identifies the conditional image distribution when the vectors are independent. \pnote{Refer to the lemmas instead of saying first and second fact.}\cans{Replaced the references to the first and second facts with references to the separate propositions.}

\pnote{This lemma contains three independent statements. Each of those should be a separate proposition.}\cans{Separated the three facts into propositions on affine concentration, rank filtering, and independent filtered images, with the existing short proofs.}
\begin{proposition}[Affine-subspace concentration]
\label{prop:affine-concentration}
Let $x\in\Fq^u$ have independent coordinates, each taking any single value with probability at most $\vartheta\le1$. For every fixed affine subspace of dimension at most $t\le u$, the probability that $x$ lies in it is at most $\vartheta^{u-t}$.
\end{proposition}

\begin{proof}
For a nonempty affine subspace of dimension $t'\le t$, choose $t'$ coordinates that determine its points uniquely. Conditional on these coordinates, membership forces each of the other $u-t'$ coordinates to a single value, which it takes with probability at most $\vartheta$. Independence gives probability at most $\vartheta^{u-t'}\le\vartheta^{u-t}$; the empty case is immediate.
\end{proof}

\begin{proposition}[Filtering and rank]
\label{prop:rank-filtering}
Fix a level $i$ and vectors $v_1,\ldots,v_s\in\Fq^d$ of rank $t$. For a uniform submatrix $A_i\in\Fq^{n_i\times d}$,
\begin{equation}
\label{eq:moment-rank-filter}
 \Prb[A_i v_j\in\{-R_i,\ldots,R_i\}^{n_i}\bmod q\text{ for all }j]\le(W_i/q)^{n_i t}=\rho_i^t.
\end{equation}
\pnote{This last sentence is hard to interpret at this point. Either state it unambiguously, or derive it from the previous statement at the point where it is used, where it would be clearer in context.}\cans{Stated the filtering event directly for $A_i$ and included the equality $(W_i/q)^{n_i t}=\rho_i^t$; the proof now explains the multiplication over independent rows.}
\end{proposition}

\begin{proposition}[Independent images conditional on filtering]
\label{prop:independent-filtered-images}
For fixed linearly independent vectors $v_1,\ldots,v_s\in\Fq^d$ and a uniform submatrix $A_i\in\Fq^{n_i\times d}$, the images $A_i v_1,\ldots,A_i v_s$ are independent and uniform on $\Fq^{n_i}$. Conditional on every image lying in $\{-R_i,\ldots,R_i\}^{n_i}\bmod q$, they are independent and uniform in this box.
\end{proposition}

\begin{proof}[Proof of \Cref{prop:rank-filtering,prop:independent-filtered-images}]
Select $t$ linearly independent vectors from the family. The linear map taking a uniform row to its inner products with those vectors is onto $\Fq^t$, so those inner products are independent uniform residues. The probability that all $t$ selected inner products lie in $\{-R_i,\ldots,R_i\}\bmod q$ is $(W_i/q)^t$; requiring the remaining inner products to lie there can only decrease it. Multiplying over the $n_i$ independent rows gives $(W_i/q)^{n_i t}=\rho_i^t$. If the whole family is linearly independent, the product uniform distribution gives independent uniform images, and conditioning it on the product filter box gives the stated conditional distribution.
\end{proof}

A linear dependency among the leaf vectors can lower the rank of the node sums at several levels. More generally, let the leaf vectors have rank deficiency $\nu$, meaning their rank is $\nu$ less than their number. As shown in the pair-survival proof below, the node sums at each level lose at most $\nu$ in rank. Each unit of rank deficiency in the node sums can increase the level-$i$ filter probability by at most $\rho_i^{-1}=Np_i\le N$, by \eqref{eq:tree-balancing} and the parameter condition $p_i\le1$ stated there. Over $r$ levels, the total survival advantage is therefore at most $N^{r\nu}$. The next lemma weights each possible rank deficiency by this survival advantage.

\begin{lemma}[Linear dependencies among leaf vectors and their survival advantage]
\label{lem:leaf-nullity}
Under \eqref{eq:moment-finite-conditions}, take any fixed collection of $1\le J\le2k$ distinct leaf entries. Let $X$ be the $d\times J$ matrix of their coefficient vectors, and let $\nu=J-\operatorname{rank}X$ be their \emph{rank deficiency}, the number of independent linear relations among these vectors. Thus $\nu=0$ exactly when the selected leaf vectors are linearly independent. Then
\begin{equation}
\label{eq:leaf-nullity-moments}
 \E[N^{r\nu}]\le1+\varepsilon,\qquad
 \E[N^{r\nu}\mathbf1_{\{\nu>0\}}]\le\frac{\varepsilon}{1-N^{-r}}\le2\varepsilon.
\end{equation}
Also $\Prb[\nu>0]\le\varepsilon/(N^r-1)\le\varepsilon$.
\end{lemma}

\begin{proof}
Reveal the selected leaf vectors in a fixed order. Let $I_j$ indicate that column $j$ belongs to the span of its predecessors. The first column has first coordinate one, so $I_1=0$. For $j\ge2$, condition on the predecessor columns and let $V$ be their span. A new column can lie in $V$ only at a point whose first coordinate is one. Since $V$ contains the first column, which has first coordinate one, fixing this coordinate removes exactly one degree of freedom from $V$. The allowed points therefore form an affine subspace of dimension $\dim(V)-1\le j-2$.

Ignoring the fixed first coordinate leaves an affine subspace of $\Fq^{d-1}$ of the same dimension. By \eqref{eq:tree-leaf-distribution}, the remaining $d-1$ coordinates of the new column are independent of one another and of the predecessor columns. Each is uniform on the $Q$ distinct residues from $[-B,B]\cap\Z$, so each residue has probability at most $1/Q$. Conditional on the predecessor columns, $I_j=1$ exactly when these coordinates lie in that fixed affine subspace. Applying \Cref{prop:affine-concentration} with $u=d-1$, $t=j-2$, and $\vartheta=1/Q$ gives
\begin{equation}
\label{eq:leaf-exposure}
 \Prb[I_j=1\mid X_{\cdot,1},\ldots,X_{\cdot,j-1}]\le Q^{-(d-j+1)}\le Q^{-D_0}.
\end{equation}
Each column lying in the span of its predecessors contributes one to $\nu$, so $\nu=\sum_j I_j$. Put $\lambda=N^r$. Since $I_j\in\{0,1\}$, we have $\lambda^{I_j}=1+(\lambda-1)I_j$, so \eqref{eq:leaf-exposure} gives $\E[\lambda^{I_j}\mid X_{\cdot,1},\ldots,X_{\cdot,j-1}]\le1+(\lambda-1)Q^{-D_0}$. The earlier indicators are determined by the conditioned columns, so iterating conditional expectation from $j=J$ down to $j=1$ gives
\begin{equation}
\label{eq:leaf-exposure-moment}
 \E[\lambda^\nu]\le\bigl(1+(\lambda-1)Q^{-D_0}\bigr)^J
 \le\exp(2kN^rQ^{-D_0})=e^{\varepsilon/2}\le1+\varepsilon.
\end{equation}
The last inequality follows from $e^x\le1+2x$ for $0\le x\le1$ and $\varepsilon\le1/N$.

To isolate the outcomes with $\nu>0$, observe that $\lambda^\nu-1$ vanishes when $\nu=0$. When $\nu>0$, we have $\lambda^\nu\ge\lambda$, so $1\le\lambda^{-1}\lambda^\nu$ and hence $\lambda^\nu-1\ge(1-\lambda^{-1})\lambda^\nu$. Taking expectations and using \eqref{eq:leaf-exposure-moment} gives
\begin{equation}
\label{eq:leaf-positive-mass}
 (1-\lambda^{-1})\E[\lambda^\nu\mathbf1_{\{\nu>0\}}]
 \le\E[\lambda^\nu-1]\le\varepsilon.
\end{equation}
This proves the weighted bound, and $(\lambda-1)\Prb[\nu>0]\le\E[\lambda^\nu-1]$ proves the probability bound, since $\lambda\ge2$.
\end{proof}

We apply \Cref{lem:leaf-nullity} to the distinct leaf entries of a full leaf selection, a selection restricted to a subtree, or a pair of selections.\pnote{This last sentence is not relevant to the lemma, is it? If not, better to put it outside.}\cans{Moved the application to full selections, subtrees, and pairs outside the lemma and its proof.}

\subsection{Shared subtrees and their weights}

For two leaf selections $a,a'\in[N]^k$, define their agreement set by $\mathsf{Sh}(a,a')=\{j\in[k]:a_j=a'_j\}$. For every $S\subseteq[k]$, let $c_i(S)$ count the level-$i$ nodes $v$ with $I(v)\subseteq S$, and let $t(S)$ count the maximal subtrees whose leaves all lie in $S$, including isolated shared leaves. Put $t(\varnothing)=0$. \pnote{Shouldn't this be in the next subsection?}\cans{Moved the shared-set definitions to the shared-subtree subsection and the rank argument to the pair-bound proof.}

Among all ordered pairs of leaf selections, a fraction $N^{-\ell}$ agree throughout a specified subtree with $\ell$ leaves. This subtree has $\ell-1$ internal nodes, so reusing its merge decisions increases joint survival by a factor $N^{\ell-1}$ when the distinct leaf vectors are linearly independent. The reference scale $N^2=N^{2k}N^{-2(k-1)}$ from Section~\ref{sec:paper-leaf-selections} combines the ordered pair count with the baseline joint survival probability. Dividing by this scale leaves the subtree's factor $N^{-\ell}N^{\ell-1}=N^{-1}$. We will count one such factor for each maximal shared subtree, then include the additional survival advantage of rank deficiency.

Maximality means that we count a shared subtree at its highest shared node, without counting its descendants again. The example in \Cref{fig:paper-shared-forest} has seven shared leaves and $t(S)=3$.

Put $p_0=N^{1-k}$. For a fixed leaf selection with linearly independent leaf vectors, this is its survival probability: the $k-1$ internal nodes independently pass their filters and are kept, each with probability $\rho_ip_i=1/N$ by \eqref{eq:tree-balancing}. The next bound combines the subtree count with the possible survival advantage from rank deficiencies among leaf vectors.

\begin{figure}[hbp]
\centering
\begin{tikzpicture}[x=.85cm,y=.6cm,tree node/.style={circle,draw,fill=white,inner sep=0pt,minimum size=3.2pt}]
  \fill[black!12,rounded corners=2pt] (.73,-.24) rectangle (1.27,.24);
  \fill[black!12,rounded corners=2pt] (1.7,-.24) -- (3.3,-.24) -- (2.5,1.32) -- cycle;
  \fill[black!12,rounded corners=2pt] (3.7,-.24) -- (7.3,-.24) -- (5.5,2.32) -- cycle;
  \foreach \name/\x/\y in {root/3.5/3,a/1.5/2,b/5.5/2,c/.5/1,d/2.5/1,e/4.5/1,f/6.5/1,l1/0/0,l2/1/0,l3/2/0,l4/3/0,l5/4/0,l6/5/0,l7/6/0,l8/7/0}
    \node[tree node] (\name) at (\x,\y) {};
  \foreach \parent/\leftchild/\rightchild in {root/a/b,a/c/d,b/e/f,c/l1/l2,d/l3/l4,e/l5/l6,f/l7/l8}
    \draw (\leftchild) -- (\parent) -- (\rightchild);
  \foreach \name in {b,d,e,f,l2,l3,l4,l5,l6,l7,l8}
    \fill (\name.center) circle (1.6pt);
  \foreach \j in {1,...,8}
    \node[below=2pt] at (l\j) {$\j$};
\end{tikzpicture}
\caption{Two leaf selections agree at positions $S=\{2,\ldots,8\}$. Their shared nodes are filled, and the three maximal shared subtrees are shaded. Agreement at seven leaf positions contributes a counting factor $N^{-7}$, while four reused merge decisions contribute a survival factor $N^4$, giving base weight $N^{-3}$. Nested shared subtrees are not counted again.}
\label{fig:paper-shared-forest}
\end{figure}

\pnote{This lemma again contains multiple independent statements that can be separated.}\cans{Separated the subtree identity, pair bound, and total forest-weight bound, stating each one's assumptions explicitly.}
\begin{proposition}[Shared-subtree identity]
\label{prop:shared-subtree-identity}
For every $S\subseteq[k]$, \pnote{What is the quantification over $S$?}\cans{Specified that the identity and pair bound hold for every $S\subseteq[k]$, including the empty and full sets.}
\begin{equation}
\label{eq:shared-subtree-identity}
 \sum_{i=1}^r c_i(S)=|S|-t(S).
\end{equation}
\end{proposition}

\begin{proof}
The maximal shared subtrees partition $S$. A subtree with $\ell$ leaves contains $\ell-1$ internal nodes, which proves \eqref{eq:shared-subtree-identity}.
\end{proof}

\begin{lemma}[Shared-subtree pair bound]
\label{lem:shared-forest}
Under \eqref{eq:moment-finite-conditions}, for every $S\subseteq[k]$,
\begin{equation}
\label{eq:weighted-pair-bound}
 \frac1{N^2}\sum_{\substack{a,a'\in[N]^k\\\mathsf{Sh}(a,a')=S}}\E[C_aC_{a'}]\le(1+\varepsilon)N^{-t(S)}.
\end{equation}
\pnote{Is the sum over $a$ and $a'$ satisfying some condition? What is the notation $S(a,a')$? Better to use a symbol different from $S$ for this}\cans{Defined the agreement-set function $\mathsf{Sh}(a,a')$ and made the domains and exact-agreement conditions of the pair sums explicit.}
\end{lemma}

\begin{proof}
Fix a pair $a,a'$ with $\mathsf{Sh}(a,a')=S$. Compare their node sums in two copies of the tree, one for each selection. Let $X$ be the $d\times J$ matrix whose columns are the coefficient vectors of the $J=2k-|S|$ distinct leaf entries, and put $\nu=J-\operatorname{rank}X$. A node sum is obtained by summing a subset of these columns. Represent that subset by its \emph{incidence vector} $h\in\Fq^J$: the zero-one vector marking its selected leaf entries. Multiplication by $X$ gives the node's coefficient vector $Xh\in\Fq^d$. At level $i$, keep the incidence vectors from both copies, except that at each of the $c_i(S)$ fully shared positions, they share the same incidence vector, so only one copy is kept.

Therefore, we have $2M_i-c_i(S)$ incidence vectors. And here we show that they are linearly independent. Distinct node positions at this level have disjoint leaf intervals. Each position that is not fully shared therefore contains, in each copy, a private leaf entry that occurs in no other kept incidence vector. A leaf entry below a fully shared position occurs only in its one kept incidence vector. Suppose they are linearly dependent, then a linear combination of these incidence vectors is zero. Nevertheless, equaling zero at these private coordinates requires that every coefficient of the linear combination is zero. \pnote{This paragraph and the next are hard to follow. Are they necessary?}\cans{The rank bound is needed for the pair-survival estimate. I have placed its proof beside that estimate, retaining the incidence-vector intuition and explaining the kernel and rank-nullity step.}

The kernel of the linear map $h\mapsto Xh$ is $\ker(X)=\{h\in\Fq^J:Xh=0\}$. By the rank-nullity theorem, its dimension is $\dim\left(\ker(X)\right)=J-\operatorname{rank}X=\nu$. Restricting this map to the span of the kept incidence vectors cannot increase its kernel dimension. If $t_i$ is the rank of the resulting level-$i$ node sums, multiplying the $2M_i-c_i(S)$ independent incidence vectors by $X$ therefore gives $t_i\ge2M_i-c_i(S)-\nu$. Equal vector values at different leaf-entry indices contribute to the rank deficiency $\nu$; we kept one copy only when the restricted leaf selections, and hence their incidence vectors, were fully shared and identical. For one leaf selection, or a leaf selection restricted to a subtree, the incidence vectors at each level have disjoint nonempty supports, and the same argument applies.

At level $i$, the pair uses exactly $2M_i-c_i(S)$ distinct subsampling variables, because it uses the same variable once at each fully shared node. Their probability of all equaling one is $p_i^{2M_i-c_i(S)}$. By \Cref{prop:rank-filtering}, conditional on $X$, the filter probability is at most $\rho_i^{t_i}\le\rho_i^{2M_i-c_i(S)-\nu}$. By \eqref{eq:tree-balancing} and the parameter condition $p_i\le1$ stated there, we have $\rho_i p_i=1/N$ and $\rho_i^{-1}=Np_i\le N$. \pnote{Where is the condition $\rho_i^{-1}=Np_i$ coming from? Every time some fact like this is used, it should have a reference to where it comes from, as there are many parameters in these arguments and it is hard to keep track of everything.}\cans{Added the balancing reference and the condition $p_i\le1$, and supplied local references at the other parameter-dependent steps.} The subsampling variables are independent of the matrix, and the submatrices for different levels are independent, so \pnote{Explain why the $p_i$ is there below, and why the second inequality is true}\cans{Explained the number of distinct subsampling coins, inserted the balancing equality, and identified the bounds used in the conditional and averaged inequalities.}
\begin{equation}
\label{eq:conditional-pair-leaf-bound}
\begin{aligned}
 \E[C_aC_{a'}\mid X]
 &\le\prod_{i=1}^r\rho_i^{2M_i-c_i(S)-\nu}p_i^{2M_i-c_i(S)}\\
 &=p_0^2N^{\sum_i c_i(S)}\prod_{i=1}^r\rho_i^{-\nu}\\
 &\le p_0^2N^{\sum_i c_i(S)}N^{r\nu}.
\end{aligned}
\end{equation}
The middle equality uses $p_0=N^{1-k}$ and $\sum_i M_i=k-1$; the last inequality uses $\rho_i^{-1}\le N$. There are exactly $N^k(N-1)^{k-|S|}\le N^{2k-|S|}$ ordered pairs with $\mathsf{Sh}(a,a')=S$. The coefficient vectors of their distinct leaf entries have the same independent product distribution, regardless of the chosen indices. Multiplying \eqref{eq:conditional-pair-leaf-bound} by this count and using $N^{2k}p_0^2=N^2$ leaves the factor $N^{-|S|+\sum_i c_i(S)}=N^{-t(S)}$, by \Cref{prop:shared-subtree-identity}. Averaging over the leaf vectors and applying \Cref{lem:leaf-nullity} proves \eqref{eq:weighted-pair-bound}.
\end{proof}

\begin{proposition}[Total shared-forest weight]
\label{prop:shared-forest-weight}
If $2k\le N$, then
\begin{equation}
\label{eq:shared-forest-bound}
 F_r:=\sum_{S\subseteq[k]}N^{-t(S)}\le e^{2k/N}\le1+4k/N.
\end{equation}
\end{proposition}

\begin{proof}
Put $x=N^{-1}$, and let $V$ be the set of all $2k-1$ nodes of the tree, leaves included. For $S\subseteq[k]$, let $T(S)\subseteq V$ be the set of roots of its maximal shared subtrees, so $|T(S)|=t(S)$. Since $S$ is the union of the leaves below $T(S)$, distinct sets $S$ give distinct sets $T(S)$. Thus each term $x^{t(S)}$ of $F_r$ is a different term of $\sum_{T\subseteq V}x^{|T|}$, and the remaining terms are nonnegative. By the binomial theorem and $1+x\le e^x$, $F_r\le\sum_{T\subseteq V}x^{|T|}=(1+x)^{2k-1}\le\exp((2k-1)/N)\le e^{2k/N}$. Apply $e^u\le1+2u$ for $0\le u\le1$, using $2k\le N$, to finish.
\end{proof}

Summing \eqref{eq:weighted-pair-bound} over nonempty $S$ and multiplying by $N^2$ bounds all shared pairs by $(1+\varepsilon)N^2(F_r-1)$; the empty set contributes $1$ to $F_r$. This includes the case $a=a'$, for which $S=[k]$ is one maximal shared subtree and $\sum_a\E[C_a]\le(1+\varepsilon)N$. The count uses indices: equal values at different leaf-entry indices remain separate and contribute to $\nu$, rather than to $t(S)$.

\subsection{Comparing the actual and ideal means, covariance, and the uniform case}
\label{sec:moment-comparison}

For a fixed leaf selection with linearly independent leaf vectors, the actual node images are independent uniform field vectors, by \Cref{prop:independent-filtered-images} and independence of the level submatrices. Conditional on survival, they are independent and uniform in their filter boxes, just as in the ideal model with those leaf vectors held fixed. We therefore compare the actual conditional mean $\E[\xi_a\mid X_a]$ with $p_0$ times the ideal mean for the same leaf vectors, before averaging over the leaf vectors. These quantities agree when the leaf vectors are linearly independent. For leaf vectors with a linear dependency, \Cref{lem:leaf-nullity} bounds their contribution to the mean discrepancy by accounting for their survival advantage.

For the variance, we separate pairs that share leaf entries from pairs that do not. The shared-subtree bound handles the former. For disjoint selections, their independently sampled leaf vectors allow us to express the covariance using only cases where the $2k$ selected vectors have a linear dependency.

\begin{proof}[Proof of \Cref{thm:reusable-tree-moments}]
\emph{The mean and list sizes.} Recall that $C_a$ indicates survival, $\xi_a=C_ag(X_a,Y_a)$ with $|g|\le1$, and $p_0=N^{1-k}$ is the survival probability for linearly independent leaf vectors. Let $\nu$ be the rank deficiency of the $k$ selected leaf vectors. The argument for \eqref{eq:conditional-pair-leaf-bound}, applied to one selection, gives \pnote{Need a reminder for some of these symbols}\cans{Recalled the survival indicator, bounded contribution, baseline probability $p_0$, and rank deficiency immediately before the estimate.}
\begin{equation}
\label{eq:single-witness-domination}
 \E[C_a\mid X_a]\le p_0N^{r\nu},\qquad
 |\E[\xi_a\mid X_a]|\le p_0N^{r\nu},
\end{equation}
If $\nu=0$, \Cref{prop:independent-filtered-images} gives independent uniform node images at each level; different levels use independent submatrices. Conditional on all filters passing, these images are therefore independent uniform filter-box vectors. Requiring the independent subsampling variables to equal one changes no image distribution. Consequently, for these fixed leaf vectors, \pnote{What is the distribution of $Y^*$? it is hard for people to remember these things from five pages ago, so reminders are needed.}\cans{Specified the distribution and independence of every ideal node image in $Y^*$ at this use.}
\begin{equation}
\label{eq:conditional-ideal-mean}
 \E[\xi_a\mid X_a]=p_0\gamma(X_a),\qquad
 \gamma(X_a)=\E_{Y^*}[g(X_a,Y^*)],\qquad |\gamma(X_a)|\le1.
\end{equation}
Here $Y^*=(Y_v^*)_{v\ \mathrm{internal}}$ has mutually independent node images, independent of the leaf vectors, with $Y_v^*$ uniform in $\{-R_i,\ldots,R_i\}^{n_i}\bmod q$ when $\operatorname{ht}(v)=i$. Thus $\gamma$ averages over these ideal images while holding its leaf-vector argument fixed. Averaging next over the sampled leaf vectors gives $\E[\gamma(X_a)]=\Gamma_g$. To compare the means, subtract $p_0\gamma(X_a)$ from $\E[\xi_a\mid X_a]$ before averaging over the leaf vectors. The difference vanishes when $\nu=0$ and has absolute value at most $p_0(N^{r\nu}+1)$ when $\nu>0$. Using \eqref{eq:single-witness-domination} and \Cref{lem:leaf-nullity}, we obtain the magnitude bound and the mean comparison
\begin{equation}
\label{eq:single-mean-transfer}
 \begin{aligned}
 |\E[\xi_a]|&\le\E[C_a]\le p_0\E[N^{r\nu}]\le p_0(1+\varepsilon),\\
 |\E[\xi_a]-p_0\Gamma_g|
 &\le p_0\bigl(\E[N^{r\nu}\mathbf1_{\{\nu>0\}}]+\Prb[\nu>0]\bigr)\le3p_0\varepsilon.
 \end{aligned}
\end{equation}
Summing over the $N^k$ leaf selections gives the mean assertion of \Cref{thm:reusable-tree-moments}, with $Z_g,\Gamma_g$ in place of $Z,\Gamma$. A leaf selection restricted to a subtree rooted at level $j$ has $2^j$ leaf vectors and $2^j-1$ internal nodes. Its conditional survival probability is at most $N^{1-2^j}N^{j\nu}\le N^{1-2^j}N^{r\nu}$. Summing over its $N^{2^j}$ possible leaf selections gives
\begin{equation}
\label{eq:upper-subtree-list-mean}
 \E[|\mathcal L_v|]\le(1+\varepsilon)N.
\end{equation}
Leaf lists have size exactly $N$, so this proves the list-size bound at every node.

\emph{Covariance of disjoint leaf selections.} Suppose $a,a'$ share no leaf entries. Their two groups of selected leaf vectors are statistically independent, although their contributions use the same matrix. Define their conditional means and joint moment by
\begin{equation}
\label{eq:conditional-covariance-terms}
 \begin{gathered}
 m_a(X_a)=\E[\xi_a\mid X_a],\qquad m_{a'}(X_{a'})=\E[\xi_{a'}\mid X_{a'}],\\
 \mathcal J_{a,a'}(X_a,X_{a'})=\E[\xi_a\overline{\xi_{a'}}\mid X_a,X_{a'}].
 \end{gathered}
\end{equation}
Regard $X_a$ and $X_{a'}$ as matrices with the selected leaf vectors as columns, and let $\nu=2k-\operatorname{rank}[X_a\ X_{a'}]$ be the rank deficiency of all $2k$ vectors from the two selections. For fixed leaf vectors with $\nu=0$, the node images from both selections are independent at each level by \Cref{prop:independent-filtered-images}; submatrices at distinct levels and the two selections' families of subsampling variables are also independent. Their contributions are therefore conditionally independent, giving
\begin{equation}
\label{eq:disjoint-pointwise-factorization}
 \mathcal J_{a,a'}(X_a,X_{a'})=m_a(X_a)\overline{m_{a'}(X_{a'})}\qquad(\nu=0).
\end{equation}
Each conditional mean is a function only of its own selected leaf vectors. Since disjoint selections use independently sampled leaf entries, $X_a$ and $X_{a'}$ are independent. Averaging over these leaf vectors therefore gives
\begin{equation}
\label{eq:unconditioned-leaf-product}
 \begin{aligned}
 \E[m_a(X_a)\overline{m_{a'}(X_{a'})}]
 &=\E[m_a(X_a)]\,\overline{\E[m_{a'}(X_{a'})]}\\
 &=\E[\xi_a]\,\overline{\E[\xi_{a'}]}.
 \end{aligned}
\end{equation}
Write $\Cov(\xi,\xi')=\E[\xi\overline{\xi'}]-\E[\xi]\,\overline{\E[\xi']}$ for the covariance. Averaging \eqref{eq:conditional-covariance-terms} gives $\E[\mathcal J_{a,a'}]=\E[\xi_a\overline{\xi_{a'}}]$. Subtracting \eqref{eq:unconditioned-leaf-product} and using \eqref{eq:disjoint-pointwise-factorization} then gives
\begin{equation}
\label{eq:disjoint-covariance-localization}
 \begin{aligned}
 \Cov(\xi_a,\xi_{a'})
 &=\E[\mathcal J_{a,a'}-m_a\overline{m_{a'}}]\\
 &=\E\bigl[(\mathcal J_{a,a'}-m_a\overline{m_{a'}})\mathbf1_{\{\nu>0\}}\bigr].
 \end{aligned}
\end{equation}
We now bound the two terms on the event $\nu>0$. By \eqref{eq:conditional-pair-leaf-bound} with $S=\varnothing$, $|\mathcal J_{a,a'}|\le p_0^2N^{r\nu}$. To bound the product of conditional means, apply the single-selection bound to each selection separately. Let $\nu_1=k-\operatorname{rank}X_a$ and $\nu_2=k-\operatorname{rank}X_{a'}$ be their individual rank deficiencies. By \eqref{eq:single-witness-domination}, $|m_a\overline{m_{a'}}|\le p_0^2N^{r(\nu_1+\nu_2)}$. We compare this weight with the pair's weight using rank subadditivity:
\begin{equation}
\label{eq:disjoint-nullity-sum}
 \nu_1+\nu_2=2k-\operatorname{rank}X_a-\operatorname{rank}X_{a'}
 \le2k-\operatorname{rank}[X_a\ X_{a'}]=\nu.
\end{equation}
Consequently, $|m_a\overline{m_{a'}}|\le p_0^2N^{r\nu}$ as well. Applying \Cref{lem:leaf-nullity} to the two terms in \eqref{eq:disjoint-covariance-localization} yields
\begin{equation}
\label{eq:disjoint-covariance-bound}
 \lvert\Cov(\xi_a,\xi_{a'})\rvert
 \le2p_0^2\E[N^{r\nu}\mathbf1_{\{\nu>0\}}]\le4p_0^2\varepsilon.
\end{equation}
There are at most $N^{2k}$ disjoint ordered pairs, so their total absolute covariance is at most $4N^2\varepsilon$.

\emph{Covariance of selections sharing leaf entries.} Here we bound both terms in the covariance separately. For the joint moment, $|\xi_a\overline{\xi_{a'}}|\le C_aC_{a'}$, so \Cref{lem:shared-forest,prop:shared-forest-weight} give
\begin{equation}
\label{eq:shared-joint-moment-sum}
 \sum_{\substack{a,a'\in[N]^k\\\mathsf{Sh}(a,a')\ne\varnothing}}|\E[\xi_a\overline{\xi_{a'}}]|
 \le(1+\varepsilon)N^2(F_r-1)\le8kN.
\end{equation}
For the product of the individual means, exactly a fraction $1-(1-1/N)^k\le k/N$ of ordered pairs share at least one leaf entry. Using \eqref{eq:single-mean-transfer} for each individual mean gives
\begin{equation}
\label{eq:shared-marginal-product-sum}
 \sum_{\substack{a,a'\in[N]^k\\\mathsf{Sh}(a,a')\ne\varnothing}}|\E[\xi_a]\,\overline{\E[\xi_{a'}]}|
 \le(1+\varepsilon)^2kN\le4kN.
\end{equation}
Expanding the variance as the sum of all ordered covariances and bounding each by its absolute value, the disjoint total of \eqref{eq:disjoint-covariance-bound} together with \eqref{eq:shared-joint-moment-sum} and \eqref{eq:shared-marginal-product-sum} now gives
\begin{equation}
\label{eq:variance-constant-accounting}
 \Var[Z_g]\le4N^2\varepsilon+8kN+4kN
 \le4N+12kN\le16kN.
\end{equation}
The last inequalities use $\varepsilon\le1/N\le1$ and $2k\le N$.

\emph{The uniform case.} Fix the matrix, leaf vectors, and entry randomness, so that $C_a$ and $z_{\mathrm{root}}(a)$ are fixed for every leaf selection. Averaging only over the independent uniform vector $b$ gives
\begin{equation}
\label{eq:null-character-orthogonality}
 \begin{aligned}
 \E_b[|Z(b)|^2]
 &=\sum_{a,a'\in[N]^k}C_aC_{a'}\E_b[\omega^{\langle z_{\mathrm{root}}(a)-z_{\mathrm{root}}(a'),b\rangle}]\\
 &=\sum_{a,a'\in[N]^k}C_aC_{a'}\mathbf1_{\{z_{\mathrm{root}}(a)=z_{\mathrm{root}}(a')\}}.
 \end{aligned}
\end{equation}
If $z_{\mathrm{root}}(a)=z_{\mathrm{root}}(a')$, the Fourier term equals one. Otherwise, some coordinate of $z_{\mathrm{root}}(a)-z_{\mathrm{root}}(a')$ is nonzero, and the corresponding uniform coordinate of $b$ makes the exponent uniform in $\Fq$. Its average is therefore $(1/q)\sum_{t=0}^{q-1}\omega^t=0$. For pairs sharing leaf entries, discard the equality indicator and use \Cref{lem:shared-forest,prop:shared-forest-weight}. Their contribution is at most $(1+\varepsilon)N^2(F_r-1)\le8kN$, including the case $a=a'$. For disjoint leaf selections, a collision between root vectors implies
\begin{equation}
\label{eq:disjoint-root-relation}
 [X_a\ X_{a'}](\underbrace{1,\ldots,1}_{k},\underbrace{-1,\ldots,-1}_{k})^\top=0.
\end{equation}
The displayed coefficient vector is nonzero on distinct indexed columns, so a collision forces $\nu>0$. The conditional pair bound \eqref{eq:conditional-pair-leaf-bound} and \Cref{lem:leaf-nullity} give at most $2p_0^2\varepsilon$ per disjoint pair, hence at most $2N^2\varepsilon\le2N$ in total. Thus $\E_{\mathrm{unif}}[|Z(b)|^2]\le8kN+2N\le10kN$. Equal values at different leaf-entry indices are included in this linear-dependency bound.
\end{proof}

\subsection{Applying the moments to decision}
\label{sec:decision-moment-application}

\paragraph{The Fourier character for LWE.} For a secret-error pair fixed independently of the matrix, the function $g$ in \eqref{eq:fixed-lwe-character-function} satisfies $Z_g=Z(b)$ and $\Gamma_g=\Gamma$, whose product formula is given in \eqref{eq:ideal-character-product}. Taking this contribution function gives the LWE mean and variance in \Cref{thm:reusable-tree-moments}; the uniform-case bound applies directly to \eqref{eq:tree-statistic}. We now use these moments to prove the decision bounds stated in Section~\ref{sec:decision-correctness}.

\begin{proof}[Proof of \Cref{lem:decision-accurate-estimate}]
For LWE, on the event $|\widetilde Z-Z|\le T/4$, returning \textsc{uniform} implies $|Z|<5T/4$. Since $|\E[Z]|\ge N|\Gamma|-3$, the threshold condition forces $|Z-\E[Z]|\ge N|\Gamma|/2$. Chebyshev's inequality bounds the probability of this deviation by $64k/(N|\Gamma|^2)$. A union bound with the accuracy-failure event, whose probability is at most $\eta$, gives \eqref{eq:decision-alternative-error}. \pnote{rejection and acceptance are not terms used in the algorithm}\cans{Replaced acceptance and rejection by the algorithm's outputs \textnormal{\textsc{LWE}} and \textnormal{\textsc{uniform}}, including the equality case in the resource theorem.}

For the uniform case, on the event $|\widetilde Z-Z|\le T/4$, returning \textsc{LWE} implies $|Z|\ge3T/4$. Markov's inequality applied to $|Z|^2$ bounds this event by $160kN/(9T^2)$. A union bound with the accuracy-failure event gives \eqref{eq:decision-null-error}.
\end{proof}

\begin{proof}[Proof of \Cref{prop:bounded-input-decision}]
Work with the first $d=\min\{m,n\}$ samples. For sufficiently large $n$, $Cn/\ln\ln q\le d\le n$, so \Cref{lem:schedule-properties} verifies the parameter checks and hypotheses of \Cref{thm:reusable-tree-moments}. For each fixed secret-error pair, that theorem gives the LWE mean and variance in \eqref{eq:reusable-moments} and the uniform-case bound in \eqref{eq:reusable-null-moment}.

\emph{Separating the scales.} The parameter choice gives $|\Gamma|=N^{-o(1)}$ and $k=N^{o(1)}$ throughout the stated parameter range. With the threshold in \eqref{eq:decision-threshold},
\begin{equation}
\label{eq:decision-threshold-separation}
 N^{3/4}\le T<2N^{3/4},\qquad
 3+5T/4\le N|\Gamma|/2
\end{equation}
for all sufficiently large $n$. \Cref{cor:decision-accuracy} gives the accuracy bound $\Prb[|\widetilde Z-Z|>T/4]\le\eta=2n^{-4}$ for the computed statistic $\widetilde Z$, both under LWE with every fixed secret-error pair and under uniform input. Applying \Cref{lem:decision-accurate-estimate}, the remaining statistical terms are $N^{-1+o(1)}$ on LWE and at most $N^{-1/2+o(1)}$ in the uniform case. Hence both errors satisfy
\begin{equation}
\label{eq:decision-uniform-rate}
 \delta_{\mathrm{dec}}=2n^{-4}+N^{-1/2+o(1)}\le3n^{-4}
\end{equation}
for sufficiently large $n$, since $L=\Theta(n/\ln\ln q)$. The same decision bounds hold throughout the stated parameter ranges because the moment, signal, and accuracy bounds do, for every fixed secret-error pair satisfying the stated bounds.
\end{proof}

\section{Proof of the time and space bounds}
\label{sec:paper-resources}

We prove \Cref{thm:bit-resources} for the algorithm of Section~\ref{sec:paper-decision-tree}: expected bit time and worst-case bit space $N^{1+o(1)}$, storage-overflow probability at most $\tfrac32n^{-4}$, and numerical error at most $T/4$ when there is no overflow. The computation has three parts: bucket queries enumerate compatible pairs, exact sampling and the storage cap control the tree, and numerical evaluation approximates the decision statistic. We then prove \Cref{prop:reduction-resources} by adding the time costs of all calls and reusing their workspace. Throughout this section, assume the theorem's parameter conditions hold.

\subsection{Enumerating compatible pairs efficiently}
\label{sec:bucket-implementation}

\pnote{These paragraphs are hard to follow because it requires the reader to remember and keep in mind a lot of terminology, symbols, constraints and relationships. See what you can do to make it better.}\cans{Recalled the child vectors, compatibility test, images, and cell width locally; stated the three bucket guarantees and shortened the traversal setup.}

At a level-$i$ merge, the two child lists contain coefficient vectors $x,y\in\Fq^d$, and $A_i$ has $n_i$ rows. A pair is compatible when every coordinate of $A_i(x+y)$ has centered magnitude at most $R_i$, as in \eqref{eq:tree-compatible-pair}. Write $u_{\mathrm L}=A_i x$, $u_{\mathrm R}=A_i y$, and recall that $W_i=2R_i+1$ is the cell width. We group right entries into the buckets of \eqref{eq:resource-buckets} and query the cells meeting $-(u_{\mathrm L})_j+\{-R_i,\ldots,R_i\}\bmod q$ in each coordinate. We prove that at most $3^{n_i}$ buckets are queried per left entry, every compatible pair is found exactly once, and the image $A_i(x+y)$ of every returned pair lies in the enlarged radius-$3R_i$ box. The last fact bounds the work spent on pairs that fail the exact filter.

Fill initial lists from left to right in sampling order, then process levels upward and merges from left to right. Number stored entries consecutively as they are inserted, including equal vectors as separate entries. We use this order later to compare capped and uncapped runs with the same randomness.

At each merge, compute and store the image of each right entry. Sort the right entries by their tuples of cell indices (their \emph{bucket keys}) and then by entry number. Each bucket occupies a consecutive range in the sorted array. For a queried key, two binary searches locate the first array position with a key at least the query and the first position with a larger key; equal positions mean the bucket is empty. Visit left entries in increasing order of entry number, computing each left image once. For each left entry, visit queried keys in lexicographic order and process the right entries in each returned range in sorted order. Generate the query tuples one at a time, and append each sum retained by subsampling to the parent list when its pair is visited.

A nonwrapping interval of $W_i$ consecutive residues meets at most two cells. For a wrapping interval, the part beginning at zero lies in the first cell; the part ending at $q-1$ meets at most the last two cells, because the final cell may be short. Thus there are at most three distinct cells per coordinate and at most $3^{n_i}$ product queries. To obtain the cell indices, split the interval at zero if it wraps. For each resulting nonwrapping segment $[a,b]$, enumerate all integer indices from $\lfloor a/W_i\rfloor$ through $\lfloor b/W_i\rfloor$, inclusive. If these quotients agree, the segment meets just one cell, and its index is included once. The two segments of a wrapping interval meet disjoint cells under $3W_i\le q$.

Every right entry compatible with $x$ lies in a queried product cell. Conversely, if $y$ is returned, in each coordinate choose a residue in the intersection of its cell with the query interval. The chosen residue's distance within that cell from $(u_{\mathrm R})_j$ is at most $W_i-1=2R_i$, while its sum with $(u_{\mathrm L})_j$ has centered magnitude at most $R_i$. Therefore $A_i(x+y)$ lies in the enlarged radius-$3R_i$ box. The width condition $3W_i\le q$ gives $3R_i<q/2$ and $6R_i+1=3W_i-2<q$, so this box has exactly $(6R_i+1)^{n_i}$ distinct residues, with no ambiguity from wrapping.

Each right entry has one bucket key, each distinct key is queried once for a fixed left entry, and each left entry is visited once. Thus every returned pair is tested exactly once, including every compatible pair. Equal vectors have distinct entries and insertion numbers. Induction on the tree shows that each possible entry is produced exactly when its required filters pass and its subsampling variables equal one, as in \eqref{eq:tree-survival}. The specified orders determine one traversal; the observed child images determine the queried keys and exact-test outcomes.

\paragraph{Expected bit time.} For this bound, consider the tree without imposing the storage cap. Let $v$ be a level-$i$ node with children $v_0,v_1$. Condition on all leaf vectors, submatrices below level $i$, and entry randomness below level $i$; call this the lower data. The child lists are then fixed and $A_i$ is still uniform. This conditioning does not include earlier queries at level $i$. By \eqref{eq:tree-anchor}, every child-pair sum has nonzero first coordinate $2^i\bmod q$, so its image under $A_i$ is uniform. It passes its exact filter and is kept by its independent subsampling variable with probability $\rho_i p_i=1/N$ by \eqref{eq:tree-balancing}. Linearity, followed by averaging over the lower data and applying \eqref{eq:upper-subtree-list-mean}, gives
\begin{equation}
\label{eq:resource-sibling-identity}
 \begin{aligned}
 \E[|\mathcal L_v|\mid\text{lower data}]&=|\mathcal L_{v_0}|\,|\mathcal L_{v_1}|/N, \qquad
 \E[|\mathcal L_{v_0}|\,|\mathcal L_{v_1}|]&=N\E[|\mathcal L_v|]\le(1+\varepsilon)N^2.
 \end{aligned}
\end{equation}
This identity bounds the expected product of the sibling list sizes despite their dependence. Write $\mathsf{Ret}_v$ for the number of candidate pairs returned by the bucket queries at $v$, before the exact compatibility check and subsampling. The enlarged box bounds each pair's conditional return probability by $((6R_i+1)/q)^{n_i}\le3^{n_i}\rho_i$, so
\begin{equation}
\label{eq:resource-returned-pairs}
 \E[\mathsf{Ret}_v]\le3^{n_i}\rho_i\E[|\mathcal L_{v_0}|\,|\mathcal L_{v_1}|]
 \le(1+\varepsilon)3^{n_i}(N\rho_i)N.
\end{equation}

The initial pass over the right list and the visit to each left entry compute one image per child entry; by \eqref{eq:upper-subtree-list-mean}, the expected number of child entries processed this way is $O(N)$ per merge. Each left entry generates at most $3^{n_i}$ bucket queries, so the expected number of queries is at most $(1+\varepsilon)3^{n_i}N$. For $J_{\mathrm R}=|\mathcal L_{v_1}|$ right entries, comparison sorting uses $O(J_{\mathrm R}\log(J_{\mathrm R}+2))$ comparisons, and each bucket lookup uses $O(\log(J_{\mathrm R}+2))$ comparisons for its two binary searches. Every tree list contains at most $N^k$ entries, so $\log(J_{\mathrm R}+2)=O(kL)$ is polynomial in $n$. Sorting therefore uses $O(NkL)$ expected comparisons per merge.

Computing an image $A_i z$ uses $O(n_i d)$ field operations per child entry. For each returned pair, add its stored images $u_{\mathrm L}+u_{\mathrm R}$ and check whether every coordinate has centered magnitude at most $R_i$, as required by \eqref{eq:tree-compatible-pair}. If the pair passes this check and is retained by subsampling, form its vector sum $x+y$ and store the sum with references to its two child entries. Computing cell indices and query-interval endpoints uses integer division and modular arithmetic. Comparing bucket keys, advancing through array positions, and reading or storing child references also take polynomially many bit operations per entry, query, or returned pair.

All these operands have polynomial bit length. Field elements use $O(\log q)$ bits, and a bucket key uses $O(n_i\log q)$ bits. Use the insertion numbers assigned above as entry references. There are $2k-1$ tree nodes and at most $N^k$ entries per list, so entry numbers use $O(kL+\log k)$ bits even without the cap; array positions satisfy the same bound. By \eqref{eq:resource-conditions}, $3^{n_i},N\rho_i=N^{o(1)}$, and $k<q$ is polynomial in $n$, so the entry, query, and returned-pair bounds above give $N^{1+o(1)}$ expected bit time over all nodes, apart from sampling, which we account for next.

\subsection{Exact sampling and the total cap}
\label{sec:resource-sampling}

Exact balancing requires Bernoulli subsampling with exact rational success probabilities. To sample uniformly from $\{0,\ldots,D-1\}$ for an integer $D\ge1$, draw $\lceil\log_2D\rceil$ fresh unbiased bits and interpret them as an integer. If it is at least $D$, discard the trial and repeat; for $D=1$, return zero. Denote this sampler by \textsc{Uniform}$(D)$. It supplies the subsampling variables, leaf vectors, and perturbation vectors used in search. Each rejected trial overwrites the same workspace, so further trials increase the running time without using more space.

Taking $D=Q=2B+1$ and subtracting $B$ gives the independent leaf coordinates $X_2,\ldots,X_d$ in \eqref{eq:tree-leaf-distribution}; taking $D=q$ independently for each of $d$ coordinates gives the uniform perturbation vector $w\in\Fq^d$ from Section~\ref{sec:paper-reductions}. A Bernoulli subsampling variable with success probability $a/D$ equals one precisely when \textsc{Uniform}$(D)<a$; it is always zero when $a=0$ and always one when $a=D$. In particular, \eqref{eq:tree-balancing} gives $a=q^{n_i}$ and $D=NW_i^{n_i}$, without rounding. Sampling the leaf entries costs $O(kNd\log q)$ expected bit time. Generating each requested Bernoulli sample costs polynomial expected bit time and worst-case workspace by the assumed polynomial encoding lengths of its numerator and denominator.

% Keep this short statement and proof together, and the link with their first line.
\begin{samepage}
The following lemma gives the conditional guarantee needed for each request.

\begin{lemma}[Exact sampling with bounded workspace]
\leavevmode\phantomsection\label{lem:resource-exact-sampling}The sampler \textup{\textsc{Uniform}}$(D)$ defined above is exactly uniform on $\{0,\ldots,D-1\}$ and terminates with probability $1$. For $D\ge2$, its expected number of trials is less than two, its expected bit time is $O(\log D)$, and its worst-case workspace is $O(\log D)$; the probability of more than $j$ trials is at most $2^{-j}$. For $D=1$ its output is always zero. These guarantees also hold conditional on the preceding computation, since each request uses fresh bits.
\end{lemma}

\begin{proof}
Put $\ell_D=\lceil\log_2D\rceil$. Each trial accepts with probability $D/2^{\ell_D}>1/2$. The output is uniform because all acceptable values are equally likely at every trial. Independent retries give probability at most $2^{-j}$ of needing more than $j$ trials, hence termination with probability $1$ and fewer than two trials on average. A trial takes $O(\ell_D)$ bit time. Store only $D$, the current trial value, and the number of bits drawn in that trial; rejected trials overwrite this workspace, and no counter of past trials or old random bits is stored. Fresh unbiased bits allow the same argument conditional on the preceding computation.
\end{proof}
\end{samepage}

\paragraph{Subsampling requests.} Each possible internal entry has a unique producing child pair, visited at most once by the traversal in Section~\ref{sec:bucket-implementation}. Generate its Bernoulli variable from fresh bits when this pair passes the compatibility check. This realizes the independent entry randomness of \eqref{eq:tree-survival} as needed. By \Cref{lem:resource-exact-sampling}, each request has polynomial expected bit time conditional on the preceding computation. Summing these conditional bounds gives total expected subsampling time at most a polynomial factor times $\sum_v\E[\mathsf{Ret}_v]$, even though earlier outcomes determine which requests are made. Here $v$ ranges over the internal nodes, each bucket candidate causes at most one request, and \eqref{eq:resource-returned-pairs} bounds each summand.

\paragraph{Comparing executions at the total cap.} Compare the algorithm with and without the storage cap, using the same random bits for all sampling and the processing orders from Section~\ref{sec:bucket-implementation}. The capped and uncapped computations agree until a retained pair would require storing entry $K+1$. All $kN\le K$ leaf entries count toward the cap, and every stored internal entry remains until the algorithm ends. If $V$ is the total number of entries produced without the cap, \eqref{eq:upper-subtree-list-mean} and the choice $K=2kn^4N$ in \eqref{eq:decision-integers} give
\begin{equation}
\label{eq:resource-cap-bound}
 \E[V]\le kN+(k-1)(1+\varepsilon)N\le3kN,\qquad
 \Prb[V>K]\le\frac{3kN}{K}=\tfrac32n^{-4}.
\end{equation}
Overflow occurs exactly when $V>K$, except on the probability-zero event of a nonterminating sampler. At $V=K$ the algorithm completes normally. On overflow, the algorithm returns \textsc{uniform} and its capped root list is defined to be empty. Otherwise every root entry is stored, counting entries with equal vector values separately. Until it stops, the capped algorithm follows the same traversal and makes the same sampling requests as the algorithm without the cap. Stopping early can only reduce these costs, so the expected construction and sampling time bounds above also apply to the capped algorithm.

\paragraph{Temporary space and termination.} Each stored entry contains its vector in $\Fq^d$, its node, and, for an internal entry, two child references. These references are the insertion numbers from Section~\ref{sec:bucket-implementation}, so each entry uses $O(d\log q+\log K+\log k)$ bits. The selected leaf entries can be recovered by following these references. There are at most $K$ stored entries.

During one merge, the sorted right array and an auxiliary comparison-sort buffer use $O(K)$ temporary items. Each item stores a right entry's bucket key, image, and entry number, all of polynomial bit length. To enumerate query tuples, keep the at most three cell indices per coordinate and one counter per coordinate recording the current choice; generate tuples one at a time in lexicographic order. The two array positions found by binary search specify the start of each returned bucket and the position just after its last entry. Process the pairs in this range as they are visited, and release the sort buffers after the merge. All counters and the active sampling workspace use polynomial bit space. Thus the algorithm uses worst-case bit space $K\operatorname{poly}(n)=N^{1+o(1)}$ during list construction, regardless of how many sampling trials are rejected.

The family of possible entries is finite, and each merge queries finitely many cells and pairs. All requested samplers terminate with probability $1$; the subsequent numerical evaluation also uses finitely many operations, so the algorithm terminates with probability $1$.

\subsection{Numerical evaluation of the Fourier statistic}
\label{sec:resource-evaluation}

Processing the root entries requires only integer arithmetic: count how many entries produce each exponent modulo $q$. Only $q$ distinct roots of unity then require approximation. An error bound for each root of unity gives an absolute error bound for the whole sum.

\begin{lemma}[Approximating roots of unity]
\label{lem:resource-certified-roots}
For $q\ge3$, $0\le t<q$, and $P\ge1$, exact rational arithmetic computes a complex number $\widetilde\omega_t\in2^{-(P+3)}(\Z+\mathrm{i}\Z)$ with a rational certificate of $|\widetilde\omega_t-e^{2\pi\mathrm{i}t/q}|\le2^{-P}$. The bit time and every intermediate bit length are polynomial in $P+\log q$.
\end{lemma}

\begin{proof}
We use Taylor series with explicit remainder bounds~\cite[Section~4.4]{BrentZimmermann10}. Set $\bar t=t$ if $t\le q/2$ and $\bar t=t-q$ otherwise, so $|\bar t|\le q/2$, and put $\eta=2^{-P-6}$. To approximate $\pi$, use the identity $\pi=4(\arctan(1/2)+\arctan(1/3))$ and truncate the alternating series for the two arctangents; their errors are bounded by the first omitted terms. This gives a rational interval of width at most $\eta$ containing $\pi$ and contained in $(3,4)$; let $\widehat\pi$ be its midpoint.

The rational angle $\widehat\theta=2\bar t\widehat\pi/q$ satisfies $|\widehat\theta|<4$ and $|\widehat\theta-2\pi\bar t/q|\le\eta/2$. The Taylor polynomials about zero for $\cos\widehat\theta$ and $\sin\widehat\theta$ through degree $J$ have error at most $4^{J+1}/(J+1)!$; choose $J$ making this bound at most $\eta$. Since changing a real angle by $\varepsilon$ changes its sine or cosine by at most $\varepsilon$, these approximations differ from the coordinates of $e^{2\pi\mathrm{i}t/q}$ by at most $3\eta/2$.

Round each rational approximation to the nearest multiple of $2^{-(P+3)}$, resolving ties by a fixed rule. Each coordinate error is at most
\begin{equation}
\label{eq:resource-root-coordinate-error}
 \tfrac32\eta+2^{-P-4}<2^{-P-3},
\end{equation}
so the complex error is less than $\sqrt2\,2^{-P-3}<2^{-P}$. The rational intervals and exact rounding errors certify this bound. Both arctangent series have geometrically decreasing error bounds, and the Taylor remainder bound decreases by at least a factor of two at each additional degree once $J\ge6$. Thus $O(P)$ terms suffice for all three approximations. Their exact fractions have polynomial bit length in $P+\log q$, so evaluation, certification, and rounding all have polynomial bit time.
\end{proof}

Write $\mathcal L_{\mathrm{root}}^{\mathrm{cap}}$ for the stored root list, defined to be empty on a storage abort, and put
\begin{equation}
\label{eq:resource-capped-statistic}
 Z^{\mathrm{cap}}(b)=\sum_{z\in\mathcal L_{\mathrm{root}}^{\mathrm{cap}}}\omega^{\langle z,b\rangle}.
\end{equation}
When list construction finishes without exceeding the storage limit, initialize the $q$ integer counters $c_0,\ldots,c_{q-1}$ to zero and process each root entry. Compute $t=\langle z,b\rangle\bmod q$ by exact modular arithmetic and increment $c_t$, giving
\begin{equation}
\label{eq:decision-histogram}
 c_t=\sum_{z\in\mathcal L_{\mathrm{root}}^{\mathrm{cap}}}\mathbf1_{\{\langle z,b\rangle=t\bmod q\}},\qquad t=0,\ldots,q-1.
\end{equation}
Consequently,
\begin{equation}
\label{eq:resource-histogram}
 Z^{\mathrm{cap}}(b)=\sum_{t=0}^{q-1}c_t\omega^t,\qquad
 0\le c_t\le K,\qquad \sum_{t=0}^{q-1}c_t\le K.
\end{equation}
Every root entry is counted separately. On abort use zero counters and $\widetilde Z=0$ without evaluation. Apply \Cref{lem:resource-certified-roots} one root of unity at a time and form $\widetilde Z=\sum_{t=0}^{q-1}c_t\widetilde\omega_t$ by multiplying and adding their integer numerators exactly at the common denominator $2^{P+3}$. The triangle inequality and the definition of $P$ give
\begin{equation}
\label{eq:resource-histogram-error}
 \left|\sum_t c_t\widetilde\omega_t-Z^{\mathrm{cap}}(b)\right|
 \le\sum_t c_t2^{-P}\le K2^{-P}\le T/4.
\end{equation}
This bound holds for every matrix, root list, and $b$.

The precision parameter $P$ in \eqref{eq:decision-integers} is the smallest integer $P\ge1$ satisfying $2^PT\ge4K$. Since $\log_2K=L+O(\log n)$ and $\log_2T=\lceil3L/4\rceil$, this choice gives $P=L/4+O(\log n)$. The $q$ counters use $q\lceil\log_2(K+1)\rceil$ bits. Since $q$ and $P$ are polynomial in $n$, the $q$ evaluations of $\omega^t$ use polynomial bit time and their reusable numerical workspace uses polynomial bit space. Write $\widetilde Z=(u+\mathrm{i}v)/2^{P+3}$, where the integer numerators $u$ and $v$ each have $O(P+\log K)$ bits. Perform the exact integer comparison
\begin{equation}
\label{eq:resource-exact-threshold}
 u^2+v^2\ge2^{2(P+3)}T^2.
\end{equation}
The histogram can be reset and its workspace reused between calls.

\begin{proof}[Proof of \Cref{thm:bit-resources}]
Enumeration of every compatible pair and fresh entry randomness give the analyzed tree process, with overflow bounded by \eqref{eq:resource-cap-bound}. Sampling and traversal have expected bit time $N^{1+o(1)}$. Reading the input and computing the parameters use polynomial bit time, since the number of original samples $m$ and the encoding lengths of the public bounds are polynomial in $n$. For every right-hand-side vector $b$, processing the root entries to form the histogram takes at most $K\operatorname{poly}(n)=N^{1+o(1)}$ worst-case bit time, and its subsequent numerical evaluation and threshold comparison take polynomial bit time. Adding their space to the stored entries and sort buffers preserves the worst-case bit-space bound $N^{1+o(1)}$. Equations \eqref{eq:resource-histogram-error} and \eqref{eq:resource-exact-threshold} give the numerical error bound for the capped statistic and the exact comparison. When list construction finishes without exceeding the storage limit, its root list equals the root list generated without the cap, so the same deterministic error bound gives $|\widetilde Z-Z(b)|\le K2^{-P}\le T/4$. The proof uses only the parameter checks in Section~\ref{sec:parameter-computation} and the assumptions of \Cref{thm:bit-resources}, including \eqref{eq:moment-finite-conditions} and \eqref{eq:resource-conditions}. The bounds therefore hold throughout any parameter range satisfying these assumptions uniformly.
\end{proof}

\subsection{Resources of the reductions}
\label{sec:reduction-resources}

The reductions use the same original samples across calls, so their runtimes can be dependent. Linearity of expectation bounds the total time without requiring independence. One tail bound then gives the runtime guarantee for the entire search computation.

\begin{proof}[Proof of \Cref{prop:reduction-resources}]
Computing the parameters and reading the input use polynomial bit time and space. For arbitrary secrets, Gaussian elimination over $\Fq$ tests the reserved samples' matrix $A_B$ for invertibility. If it is invertible, the resulting matrix factorization is reused to form \eqref{eq:reduction-normal-form} and, after bounded-secret search, to perform the final linear solve \eqref{eq:reduction-recovery}. All these operations take polynomial bit time. Each perturbation samples the $d$ coordinates of $w$ exactly using \Cref{lem:resource-exact-sampling}, then performs polynomially many bit operations. Sampling the perturbation vector $w$ takes polynomial expected bit time, and rejected trials reuse the same bounded workspace.

After successful preprocessing, decision makes one call to the bounded-secret decision algorithm, while search makes $M_{\mathrm{calls}}=n^{1+o(1)}$ such calls with identical algorithm parameters, by \eqref{eq:reduction-guess-count}. For bounded secrets, the input matrix and each perturbed matrix used in search are marginally uniform. For arbitrary secrets, the normal-form matrix and each subsequent perturbed matrix are uniform conditional on any fixed invertible $A_B$. The bijections in Appendix~\ref{app:reduction-proofs} establish both statements.

Each call constructs the coefficient-vector lists of \Cref{alg:bounded-input-decision} using its input matrix and fresh randomness; evaluation of the Fourier statistic has a worst-case time bound for every right-hand-side vector, by \Cref{thm:bit-resources}. We can therefore sum that theorem's one-call expected bounds, using the parameters of \Cref{lem:schedule-properties}, without conditioning on earlier answers. In the arbitrary-secret case, these bounds hold conditional on each fixed invertible $A_B$; if $A_B$ is singular, the algorithm stops after the polynomial-time invertibility test. If $\mathsf W_{\mathrm{total}}$ denotes all bit time, including preprocessing, perturbations, calls, and recovery, linearity gives
\begin{equation}
\label{eq:reduction-total-work}
 \E[\mathsf W_{\mathrm{total}}]
 \le\operatorname{poly}(n)+n^{1+o(1)}2^{(1/2+o(1))n/\ln\ln q}
 =2^{(1/2+o(1))n/\ln\ln q}.
\end{equation}
One Markov bound now yields
\begin{equation}
\label{eq:reduction-runtime-tail}
 \Prb[\mathsf W_{\mathrm{total}}>n\E[\mathsf W_{\mathrm{total}}]]\le1/n.
\end{equation}
The factor $n$ preserves the exponent because $\log n=o(n/\ln\ln q)$. We use this threshold only to bound total runtime; the algorithm continues to completion, so the runtime bound adds no correctness failure.

Execute the calls sequentially, releasing their stored entries and reusing the workspace for evaluating the Fourier statistic. Keep only the selected samples, the reserved samples' matrix factorization when needed, the current perturbed pair, and, for each coordinate, one accepted guess and an acceptance count capped at two. This adds polynomial space to one call's worst-case bound from \Cref{thm:bit-resources}. The sampler for perturbation coordinates also overwrites rejected trials without storing a trial counter. Every call terminates with probability $1$, so the finite sequence does too; arbitrarily many rejected sampling trials increase the running time but reuse the same workspace, preserving the worst-case space bound.
\end{proof}

\section{Parameter choices}
\label{sec:paper-parameters}

We now choose the leaf width $Q$, filter widths $W_i$, and submatrix row counts $n_i$ so that the moment and resource conditions hold and the bias $|\Gamma|$ is large enough for decision. Our targets from \Cref{lem:schedule-properties} are $L=(1/2+o(1))n/\ln\ln q$, $k=N^{o(1)}$, and $|\Gamma|=N^{-o(1)}$. The two requirements compete: wider leaf intervals make linear dependencies rarer, while wider filter intervals reduce the per-row filtering strengths, so each level can filter more rows; both increase the \emph{bias loss} $-\ln|\Gamma|$. We verify \eqref{eq:moment-finite-conditions} and \eqref{eq:resource-conditions} and compute the mean in the ideal model for one leaf selection; \Cref{thm:reusable-tree-moments} then bounds $|\E[Z]-N\Gamma|$ for the algorithm run without a storage limit.

\subsection{Computing the parameters}
\label{sec:parameter-computation}

\paragraph{Choosing the filter intervals.} We choose the filter radius $R_i$ on the scale $q/(f\sqrt{M_i})$, with scale factor $f$ fixed below and interval width $W_i=2R_i+1$. In the ideal model, each coordinate of the sum of the $M_i$ node images at level $i$ has variance $M_i(W_i^2-1)/12$ before reduction modulo $q$: the summands are independent and uniform on $[-R_i,R_i]\cap\Z$. This choice keeps the variance on the scale $q^2/f^2$ across levels and controls the secret's contribution to the bias loss $-\ln|\Gamma|$.

We round $q/(f\sqrt{M_i})$ down to obtain integer $R_i$ and balance the row counts $n_i$ to limit list growth before subsampling. \Cref{lem:schedule-properties} proves the bit-time, linear-dependency, and signal bounds for these integer choices.

\paragraph{Public inputs and parameter checks.} Let $n\ge1$, let $q$ be prime, and let $E_s,E_e\ge0$ be rational public bounds with polynomially many bits in $n$. The number of original samples $m$ is polynomial in $n$, so reading the input costs polynomially many bit operations in the modulus regimes below. The bounded-secret decision algorithm uses the first $d=\min\{m,n\}$ samples; normal form uses the first $d=\min\{m-n,n\}$ samples and reserves the final $n$ samples when $m>n$, and supplies $E_s=E_e=E$ by \Cref{prop:fixed-suffix-normal-form}. Each decision call therefore has $d\le n$. All parameter checks below pass for sufficiently large $n$ in the range of \Cref{lem:schedule-properties}.

For a positive integer $u$, let $\operatorname{bitlen}(u)=1+\lfloor\log_2u\rfloor$ denote its bit length. Define
\begin{equation}
\label{eq:schedule-scales}
 \begin{gathered}
 h=\operatorname{bitlen}(2n-1)=\lceil\log_2(2n)\rceil,\qquad f=32\overline E_s h,\\
 \overline E_s=\max\{1,\lceil E_s\rceil\},\qquad \overline E_e=\max\{1,\lceil E_e\rceil\}.
 \end{gathered}
\end{equation}
Check that $\min\{d,q\}\ge2h$. Choose the largest power of two $k\le\min\{d,q\}/h$, and put $r=\log_2k$. For levels $1\le i\le r$, set
\begin{equation}
\label{eq:schedule-intervals}
 \begin{gathered}
 B=\left\lfloor\frac{\sqrt q}{2\overline E_e}\right\rfloor,\qquad Q=2B+1,\qquad M_i=k/2^i,\\
 R_i=\left\lfloor\frac q{f\sqrt{M_i}}\right\rfloor,\qquad W_i=2R_i+1,
 \end{gathered}
\end{equation}
Check that $B\ge1$, $Q<q$, and $R_i\ge1$, $3W_i\le q$ at every level. The coefficient vector of each leaf entry is sampled independently as
\begin{equation}
\label{eq:schedule-leaves}
 (1,X_2,\ldots,X_d)\bmod q,\qquad X_2,\ldots,X_d\text{ independent and uniform on }[-B,B]\cap\Z.
\end{equation}
All vector coordinates are stored as residues modulo $q$.

\paragraph{Choosing the row counts.} Write $\ell_i=\log_2(q/W_i)$ for the per-row filtering strength. Since $\rho_i=(W_i/q)^{n_i}=2^{-n_i\ell_i}$, we balance the total strengths $n_i\ell_i$ near $L$, targeting a filter probability $1/N$ at every level. Ignoring the need for integer row counts, this allocates $L/\ell_i$ rows to level $i$. Thus $H=\sum_{i=1}^r\ell_i^{-1}$ gives a total of $LH$ rows, suggesting the list-size exponent $L\approx n/H$. These logarithms are for analysis only; the algorithm uses the integer comparison \eqref{eq:schedule-exact-comparison} to compute the row counts exactly. Let
\begin{equation}
\label{eq:schedule-capacities}
 L=\min\left\{l\in\Z_{\ge1}:\sum_{i=1}^r\left\lfloor\frac l{\ell_i}\right\rfloor\ge n\right\},\qquad
 N=2^L,\qquad b_i=\left\lfloor\frac L{\ell_i}\right\rfloor.
\end{equation}
We call $b_i$ the \emph{row capacity} at level $i$: the largest row count with total filtering strength at most $L$. Write $\Delta=\sum_i b_i-n$ for the excess capacity, and subtract one from the capacities at the $\Delta$ lowest levels. This gives row counts $n_1,\ldots,n_r$ summing to $n$. The bound $0\le\Delta\le r-1$, proved below, ensures there are enough distinct levels for these subtractions. Check that every $n_i$ is positive, and use these sizes for consecutive submatrices $A_i$. With
\begin{equation}
\label{eq:schedule-coins}
 \rho_i=(W_i/q)^{n_i},\qquad p_i=\frac{q^{n_i}}{NW_i^{n_i}}=\frac1{N\rho_i},
\end{equation}
a uniform level-$i$ image passes its filter and is kept by the independent subsampling variable with probability $\rho_i p_i=1/N$. If any parameter check fails, the bounded-secret decision algorithm outputs ``uniform'' and search reports failure before its calls. These checks concern only the number of lists, leaf and filter intervals, and positive row counts, and depend only on $n,q,d,E_s,E_e$. The asymptotic inequalities below follow from the parameter choice.

Table~\ref{tab:parameter-choices} collects these choices together with the threshold, storage limit, and precision from \eqref{eq:decision-threshold} and \eqref{eq:decision-integers}.

\begin{table}[htbp]
\centering
\small
\caption{Parameter choices for levels $1\le i\le r$. The row capacities are computed by exact integer comparisons; $\ell_i,H,D_0,\varepsilon$ are analysis quantities.}
\label{tab:parameter-choices}
\begin{tabular}{@{}>{\raggedright\arraybackslash}p{.20\textwidth}>{\raggedright\arraybackslash}p{.76\textwidth}@{}}
\toprule
\textbf{Parameters} & \textbf{Values}\\
\midrule
$\overline E_s,\overline E_e$ & $\overline E_s=\max\{1,\lceil E_s\rceil\}$, $\overline E_e=\max\{1,\lceil E_e\rceil\}$.\\
$h,f$ & $h=\operatorname{bitlen}(2n-1)$, $f=32\overline E_s h$.\\
$k,r,M_i$ & $k$ is the largest power of two at most $\min\{d,q\}/h$; $r=\log_2k$, $M_i=k/2^i$.\\
$B,Q$ & $B=\lfloor\sqrt q/(2\overline E_e)\rfloor$, $Q=2B+1$.\\
$R_i,W_i$ & $R_i=\lfloor q/(f\sqrt{M_i})\rfloor$, $W_i=2R_i+1$.\\
$\ell_i,H$ & $\ell_i=\log_2(q/W_i)$, $H=\sum_{i=1}^r\ell_i^{-1}$.\\
$L,N$ & $L=\min\{l\in\Z_{\ge1}:\sum_{i=1}^r\lfloor l/\ell_i\rfloor\ge n\}$, $N=2^L$.\\
$b_i,\Delta,n_i$ & $b_i=\lfloor L/\ell_i\rfloor$, $\Delta=\sum_i b_i-n$, $n_i=b_i-\mathbf1_{\{i\le\Delta\}}$.\\
$\rho_i,p_i$ & $\rho_i=(W_i/q)^{n_i}$, $p_i=1/(N\rho_i)$.\\
$T,K$ & $T=2^{\lceil3L/4\rceil}$, $K=2kn^4N$.\\
$P$ & $P=\min\{j\in\Z_{\ge1}:2^jT\ge4K\}$.\\
$D_0,\varepsilon$ & $D_0=d-1-2k$, $\varepsilon=4kN^rQ^{-D_0}$.\\
\bottomrule
\end{tabular}
\end{table}

\FloatBarrier
\subsection{The mean in the ideal model}

The independent coordinates of the ideal model let us write its ideal mean as a product of Fourier coefficients. For an odd width $W$, define the Fourier coefficient at frequency $t$ of the uniform distribution on the centered interval of width $W$ by
\begin{equation}
\label{eq:ideal-interval-character}
 \chi_W(t)=\frac1W\sum_{u=-(W-1)/2}^{(W-1)/2}\omega^{tu},\qquad \omega=\exp(2\pi\mathrm{i}/q).
\end{equation}
Use the ideal model for one leaf selection from Section~\ref{sec:decision-correctness}. Its $k$ leaf vectors are sampled independently from the distribution in \eqref{eq:schedule-leaves}. Its node images $Y^*_v$ are independent of the leaf vectors and of one another, and each is uniform in its level's filter box, with independent coordinates. Let $z^*$ be the sum of the leaf vectors, and let the $i$th block of $y^*$ be $\sum_{\operatorname{ht}(v)=i}Y^*_v$, with both sums reduced modulo $q$. For fixed integer vectors $e\in\Z^d$, $s\in\Z^n$, use the consecutive secret blocks $s^{(i)}$ and define
\begin{equation}
\label{eq:ideal-character-product}
 \begin{gathered}
 \Gamma=\E\bigl[\omega^{\langle z^*,e\rangle+\langle y^*,s\rangle}\bigr]=\omega^{ke_1}G, \qquad
 G=\prod_{j=2}^d\chi_Q(e_j)^k\prod_{i=1}^r\prod_{j=1}^{n_i}\chi_{W_i}(s^{(i)}_j)^{M_i}.
 \end{gathered}
\end{equation}
This is the mean of the fixed Fourier character for LWE in \eqref{eq:fixed-lwe-character-function} under the ideal model. Symmetry of the centered intervals makes $G$ real. Independence of the free leaf-vector coordinates and node-image coordinates gives the product. The fixed first coordinates contribute the common factor $\omega^{ke_1}$, which changes the mean's direction but not its magnitude.

\subsection{Verifying the guarantees}
\label{sec:parameter-verification}

\pnote{We should try to see, perhaps for the next version of the paper, if we can simplify this section, even at the cost worse constants in the exponent of the running time.}\hnote{I agree that a simpler parameter choice would be useful in a later version. For now I propose retaining the current choices to preserve the leading running-time constant $1/2$, while adding the parameter table and clearer references.} \pnote{Add a table summarising the values chosen for all the parameters. Then, in earlier sections, whenever we refer to ``parameters set in section 5'', we can refer to that table instead.}\cans{Added a compact table of the chosen parameter values and directed earlier references to those values to the table, retaining references to the proofs of their guarantees.}

We prove \Cref{lem:schedule-properties}. The argument establishes valid intervals, verifies the row balance that yields the running-time exponent, and bounds the weighted contribution of linear dependencies and the bias loss $-\ln|\Gamma|$. For each fixed asymptotic growth of the public bounds, we establish all estimates uniformly in the retained dimension and secret-error pair throughout \eqref{eq:schedule-uniform-range}. We finish by verifying exact computation and bit lengths.

\begin{proof}[Proof of \Cref{lem:schedule-properties}]
\emph{Number of lists and interval sizes.} Since $\ln\ln q\sim\ln\ln n$, the full interval \eqref{eq:schedule-uniform-range} gives $\ln d=\ln n+O(\ln\ln\ln n)$ throughout the dimension range. By \eqref{eq:schedule-scales}, $h=\Theta(\ln n)$ and $\ln f=(\ln n)^{o(1)}$, so $f=n^{o(1)}$. Thus $\min\{d,q\}/h\to\infty$, which verifies $\min\{d,q\}\ge2h$ from Section~\ref{sec:parameter-computation} for sufficiently large $n$. Rounding $\min\{d,q\}/h$ down to a power of two gives
\begin{equation}
\label{eq:schedule-arity-bounds}
 \frac{\min\{d,q\}}{2h}<k\le\frac{\min\{d,q\}}h,\qquad \frac kd\le\frac1h,\qquad \frac kq\le\frac1h.
\end{equation}
These inequalities give $k=o(d)$ and $\ln k=(\min\{1,\kappa\}+o(1))\ln n$, with $r=\Theta(\ln n)$.

For the target filter radius $R_i^*=q/(f\sqrt{M_i})$, the definition in \eqref{eq:schedule-intervals} gives $R_i=\lfloor R_i^*\rfloor$. The minimum of $R_i^*$ occurs at $i=1$, and
\begin{equation}
\label{eq:schedule-radius-growth}
 R_1^*\ge\frac{q}{f\sqrt k}\ge\frac{\sqrt{qh}}f\longrightarrow\infty.
\end{equation}
For all sufficiently large $n$, the definition in \eqref{eq:schedule-intervals} gives, at every level,
\begin{equation}
\label{eq:schedule-filter-widths}
 R_i^*/2\le R_i\le R_i^*,\qquad R_i^*\le W_i\le3R_i^*.
\end{equation}
Hence every $R_i\ge1$. For the width check $3W_i\le q$, we use $1\le R_i\le q/f$ to obtain
\begin{equation}
\label{eq:schedule-no-wrap}
 3W_i=6R_i+3\le6q/f+3\le9q/f<q,
\end{equation}
because $f\ge32$ by \eqref{eq:schedule-scales}. This verifies the filter-interval checks in Section~\ref{sec:parameter-computation}. The bound $3W_i\le q$ also implies $\ell_i\ge\log_2 3>1$, $6R_i+1=3W_i-2<q$, and $3R_i<q/2$. The last two inequalities ensure that the enlarged boxes used by bucket enumeration in Section~\ref{sec:bucket-implementation} contain distinct centered residues.

Since $\overline E_e=n^{o(1)}$, \eqref{eq:schedule-intervals} gives $B\to\infty$ and $Q=O(\sqrt q)=o(q)$. Thus $B\ge1$ and $Q<q$ for sufficiently large $n$, verifying the leaf-interval checks in Section~\ref{sec:parameter-computation}. The same definition yields
\begin{equation}
\label{eq:schedule-leaf-width}
 \begin{gathered}
 \ln Q=\tfrac12\ln q-\ln\overline E_e+O(1)=\left(\tfrac12+o(1)\right)\ln q, \qquad
 QE_e\le\sqrt q+\overline E_e\le\tfrac32\sqrt q
 \end{gathered}
\end{equation}
for sufficiently large $n$. In particular $QE_e\le q/2$. The filter widths $W_i\le3R_i^*$ give $W_iE_s\le3qE_s/(f\sqrt{M_i})\le3q/(32h)<q/2$. Thus the geometry needed for enumeration and signal evaluation is
\begin{equation}
\label{eq:schedule-geometry}
 6R_i+1<q,\qquad 3R_i<q/2,\qquad QE_e\le q/2,\qquad W_iE_s\le q/2.
\end{equation}
If either of $E_s, E_e$ is zero, the corresponding product $W_iE_s$ or $QE_e$ is zero, so the corresponding inequality holds automatically; the conventions $\overline E_s,\overline E_e\ge1$ still define $f$ and $B$.

\emph{Estimating the reciprocal sum $H$.} Put $\alpha=\log_2 f$; \eqref{eq:schedule-scales} gives $\log_2(32h)\le\alpha=(\ln n)^{o(1)}$. Taking logarithms in \eqref{eq:schedule-filter-widths} gives $\ell_i=\alpha+\frac12\log_2M_i+O(1)$. Thus $\ell_{\min}:=\min_i\ell_i\to\infty$ and $\ell_{\max}:=\max_i\ell_i=O(\ln n)$. Since $M_i=2^{r-i}$, put $j=r-i$, which ranges over $0,\ldots,r-1$. The reciprocal error is
\begin{equation}
\label{eq:schedule-reciprocal-error}
 \left|\frac1{\ell_i}-\frac1{\alpha+j/2}\right|
 =\frac{|\ell_i-(\alpha+j/2)|}{\ell_i(\alpha+j/2)}
 =O\!\left((\alpha+j/2)^{-2}\right).
\end{equation}
The numerator is $O(1)$, and each denominator factor has order $\alpha+j/2$ because $\alpha\to\infty$. Summing \eqref{eq:schedule-reciprocal-error} compares $H$ with $\sum_{j=0}^{r-1}(\alpha+j/2)^{-1}$. The total error is $O(1/\alpha)$, since the sum of $(\alpha+j/2)^{-2}$ is at most $\alpha^{-2}+2/\alpha$ by comparison with its integral. For the positive decreasing function $(\alpha+x/2)^{-1}$, the integral over $[0,r]$ is a lower bound on the sum. The terms with $1\le j\le r-1$ are each bounded by the integral over $[j-1,j]$, so their sum is at most the integral over $[0,r-1]$, and hence over $[0,r]$. Adding the $j=0$ term $1/\alpha$ therefore bounds the full sum by $1/\alpha$ plus the integral over $[0,r]$. Hence
\begin{equation}
\label{eq:schedule-harmonic-integral}
 H=\sum_{j=0}^{r-1}\frac1{\alpha+j/2}+O(1/\alpha)=2\ln\!\left(\frac{\alpha+r/2}{\alpha}\right)+O(1/\alpha)=(2+o(1))\ln\ln q.
\end{equation}
Indeed $r=\Theta(\ln n)$ and $\alpha=(\ln n)^{o(1)}$, so $\ln(\alpha+r/2)=\ln\ln n+O(1)$ whereas $\ln\alpha=o(\ln\ln n)=o(\ln\ln q)$.

\emph{Exact capacities and positive row counts.} The interval checks give $1<\ell_i<\log_2q$. For a candidate integer $l$ in \eqref{eq:schedule-capacities}, we need the capacities $\lfloor l/\ell_i\rfloor$. Computing this floor from an approximate logarithm could change the row count near an integer boundary. Instead, the capacity is the largest integer $j\in[0,l]$ satisfying $j\ell_i\le l$. Exponentiating in base two gives the equivalent comparisons
\begin{equation}
\label{eq:schedule-exact-comparison}
 j\ell_i\le l\quad\Longleftrightarrow\quad(q/W_i)^j\le2^l\quad\Longleftrightarrow\quad q^j\le2^lW_i^j.
\end{equation}
The final inequality uses only integer powers. Use binary search over $j$ to compute each capacity for a fixed $l$, and over $l$ to find the smallest value whose total capacity is at least $n$. The integer bound $U_{\max}=n\operatorname{bitlen}(q-1)=n\lceil\log_2q\rceil$ is an upper search endpoint: since $\ell_i\le\lceil\log_2q\rceil$, at $l=U_{\max}$ even one capacity is at least $n$.

Since $\ell_i>1$, increasing $l$ by one raises each capacity by at most one. Minimality of $L$ in \eqref{eq:schedule-capacities} gives $\sum_i\lfloor(L-1)/\ell_i\rfloor\le n-1$, so $n\le\sum_i b_i\le n+r-1$. This proves $0\le\Delta\le r-1$.

To quantify the effect of integer row counts, we compare $L$ with $n/H$. By \eqref{eq:schedule-capacities}, $n\le\sum_i b_i\le LH$, so $L\ge n/H$. For an upper bound, set $l_*=\lceil n/H+\ell_{\max}\rceil$. The extra $\ell_{\max}$ allows at least one more row at every level, compensating for the loss of less than one from taking a floor: since $\ell_{\max}\ge\ell_i$,
\begin{equation}
\label{eq:schedule-capacity-upper}
 \left\lfloor\frac{l_*}{\ell_i}\right\rfloor
 \ge\frac{l_*}{\ell_i}-1
 \ge\frac n{H\ell_i}.
\end{equation}
Summing \eqref{eq:schedule-capacity-upper} gives total capacity at least $n$, so $l_*$ is successful in \eqref{eq:schedule-capacities}. Thus $L\le l_*\le n/H+\ell_{\max}+1$.

For the actual row counts, $b_i=\lfloor L/\ell_i\rfloor$ in \eqref{eq:schedule-capacities} and the rule $n_i\in\{b_i-1,b_i\}$ give $L/\ell_i-2\le n_i\le L/\ell_i$. Multiplying by $\ell_i$ and combining with the bounds on $L$ proves
\begin{equation}
\label{eq:schedule-rounding}
 L-2\ell_i\le n_i\ell_i\le L,\qquad \frac nH\le L\le\frac nH+\ell_{\max}+1.
\end{equation}
Since $\ell_{\max}=O(\ln n)$, these bounds and \eqref{eq:schedule-harmonic-integral} give
\begin{equation}
\label{eq:schedule-balanced-size}
 L=\frac nH+O(\ln n)=\left(\tfrac12+o(1)\right)\frac n{\ln\ln q}.
\end{equation}
Every capacity satisfies $b_i\ge L/\ell_{\max}-1\to\infty$, so $n_i\ge b_i-1>0$ for sufficiently large $n$, verifying the row-count check in Section~\ref{sec:parameter-computation}. The construction gives $\sum_i n_i=n$. Applying \eqref{eq:schedule-rounding} yields
\begin{equation}
\label{eq:schedule-balance-loss}
 0\le\log_2(N\rho_i)=L-n_i\ell_i\le2\ell_{\max}=o(L),\qquad \frac{n_i}{L}\le\frac1{\ell_{\min}}=o(1).
\end{equation}
This proves both overhead bounds in \eqref{eq:resource-conditions} and $0<p_i\le1$. Since $k\le n/h$ is polynomial and $L=\Theta(n/\ln\ln q)$, we also have $\log n=o(L)$, $k=N^{o(1)}$, and $2k\le N$ for sufficiently large $n$.

\emph{Linear-dependency bounds throughout the parameter range.} Recall $D_0=d-1-2k$ from \eqref{eq:moment-finite-conditions}. By \eqref{eq:schedule-arity-bounds}, $D_0\ge d(1-2/h)-1$. The lower bound $d\ge Cn/\ln\ln q$ in \eqref{eq:schedule-uniform-range} and the value of $L$ in \eqref{eq:schedule-balanced-size} give $D_0/L\ge2C-o(1)$ and hence $D_0\ge1$ for sufficiently large $n$. Since $k\le\min\{d,q\}/h\le\min\{n,q\}$, we have $\ln k\le(\min\{1,1/\kappa\}+o(1))\ln q$. Together with \eqref{eq:schedule-leaf-width}, this gives
\begin{equation}
\label{eq:schedule-closure-margin}
 \frac{D_0\ln Q-L\ln k}{L\ln q}\ge C-\min\{1,1/\kappa\}-o(1).
\end{equation}
\eqref{eq:schedule-closure-margin} gives $D_0\ln Q\ge L\ln k+(C-\min\{1,1/\kappa\}-o(1))L\ln q$. To prove the target bound $D_0\ln Q\ge\ln(4kN^{r+1})$, the margin beyond $L\ln k$ must also cover $L\ln2+\ln(4k)$:
\begin{equation}
\label{eq:schedule-closure-remainder}
 \ln(4kN^{r+1})=L\ln k+L\ln2+\ln(4k),\qquad
 \frac{L\ln2+\ln(4k)}{L\ln q}=o(1).
\end{equation}
Since $C>\min\{1,1/\kappa\}$, the margin exceeds these additional terms for sufficiently large $n$. Exponentiating the resulting inequality and using $\varepsilon$ from \eqref{eq:moment-finite-conditions} gives
\begin{equation}
\label{eq:schedule-finite-closure}
 4kN^{r+1}\le Q^{D_0},\qquad \varepsilon=4kN^rQ^{-D_0}\le N^{-1}.
\end{equation}
Together with $2k\le N$ and $D_0\ge1$ established above, this verifies \eqref{eq:moment-finite-conditions}.

\emph{Bounding the bias loss.} To prove the signal bound $|\Gamma|=N^{-o(1)}$ in \eqref{eq:schedule-exponent}, we show that $-\ln|\Gamma|=o(L)$. We first show that the Fourier coefficients in \eqref{eq:ideal-character-product} are positive and bound their logarithmic losses. These losses add with weights $k$ and $M_i$ to give $-\ln|\Gamma|$.

If $U_W$ is uniform on the centered interval of odd width $W$, symmetry gives $\chi_W(t)=\E[\cos(2\pi tU_W/q)]$ and $\E[U_W^2]=(W^2-1)/12$. Hence $\cos x\ge1-x^2/2$ implies
\begin{equation}
\label{eq:schedule-character-lower}
 1\ge\chi_W(t)\ge1-\frac{\pi^2(W^2-1)t^2}{6q^2}\ge1-\frac{\pi^2W^2t^2}{6q^2}.
\end{equation}
When $W|t|\le q/2$, the last subtracted quantity is at most $\pi^2/24<1/2$. Positivity and $-\ln(1-x)\le2x$ for $0\le x\le1/2$ therefore give the complete bound
\begin{equation}
\label{eq:schedule-character-loss}
 0<\chi_W(t)\le1,\qquad 0\le-\ln\chi_W(t)\le\frac{\pi^2W^2t^2}{3q^2}.
\end{equation}
At $t=0$, the Fourier coefficient is one and its logarithmic loss is zero. Applying \eqref{eq:schedule-character-loss} to each factor in \eqref{eq:ideal-character-product}, with $QE_e\le q/2$ and $W_iE_s\le q/2$ from \eqref{eq:schedule-geometry}, proves $G>0$. Since $|\Gamma|=G$, summing the logarithmic losses yields
\begin{equation}
\label{eq:schedule-product-loss}
 0\le-\ln G\le\frac{\pi^2}{3q^2}\left(kQ^2\sum_{j=2}^de_j^2+\sum_{i=1}^rM_iW_i^2\|s^{(i)}\|_2^2\right).
\end{equation}
By \eqref{eq:schedule-leaf-width} and \eqref{eq:schedule-filter-widths}, $Q^2E_e^2\le9q/4$ and $M_iW_i^2\le9q^2/f^2$ for sufficiently large $n$. The coordinate bounds give $\sum_{j=2}^d e_j^2\le dE_e^2$ and $\|s^{(i)}\|_2^2\le n_iE_s^2$. Substituting these inequalities into \eqref{eq:schedule-product-loss} bounds the error contribution by $O(dk/q)$ and the secret contribution by $O(E_s^2\sum_i n_i/f^2)=O(nE_s^2/f^2)$, using $\sum_i n_i=n$. Finally, $d\le n$ from \eqref{eq:schedule-uniform-range}, $k/q\le1/h$ from \eqref{eq:schedule-arity-bounds}, and $E_s/f\le1/(32h)$ from \eqref{eq:schedule-scales} give $-\ln G=O(n/h)$. Together with \eqref{eq:schedule-balanced-size}, this gives $(-\ln G)/L=O(\ln\ln q/h)=o(1)$. Thus
\begin{equation}
\label{eq:schedule-signal}
 0\le-\ln G=O\!\left(\frac{dk}{q}+\frac{nE_s^2}{f^2}\right)=O(n/h)=o(L),\qquad |\Gamma|=G=N^{-o(1)}.
\end{equation}
% Here $-\ln G=-\ln|\Gamma|$ is the bias loss.
% If $E_e=0$ or $E_s=0$, the respective product in \eqref{eq:ideal-character-product} contributes zero to this loss; if both vanish, $G=1$ exactly.
Together with the bounds on $L$ and $k$, this verifies \eqref{eq:schedule-exponent}.

\emph{Bit lengths, encoding lengths, and termination.} Computing the ceilings of the rational bounds, bit lengths, powers of two, and integer divisions uses polynomial bit time. The floors in \eqref{eq:schedule-intervals} are computed exactly by binary search: $B$ is the largest integer $b\in[0,q]$ with $4\overline E_e^2b^2\le q$, and $R_i$ is the largest integer $u\in[0,q]$ with $f^2M_i u^2\le q^2$. These comparisons use integers with $O(\log n)$ bits in the stated regimes, so the searches take polynomial bit time. The search bound for row capacities is $U_{\max}=O(n\log q)$; capacity comparisons have $j,l\le U_{\max}$, so their operands in \eqref{eq:schedule-exact-comparison} have $O(n(\log q)^2)$ bits. The two nested binary searches make $O(r(\log U_{\max})^2)$ such comparisons, with powers obtained by integer multiplication. Thus the capacity search takes polynomial bit time. When the parameter checks pass, $n_i\le n$, and the numerator and denominator of each subsampling probability have $O(L+n_i\log q)$ bits. The integer $N$ has $L+1$ bits, despite its exponential value.

For \eqref{eq:decision-integers}, $\log_2K=L+O(\log n)$ and $\log_2T=\lceil3L/4\rceil$. The precision parameter $P$ can be obtained by comparing integers or by subtracting this latter exponent from $\operatorname{bitlen}(4K-1)$ and taking the maximum with one. Therefore $P=L/4+O(\log n)$. There are $q$ histogram counters $c_t$, each at most $K$, and exponents have $O(\log q)$ bits. Fixed-point approximations of unit complex numbers can be represented at common denominator $2^{P+3}$ with $O(P+1)$ numerator bits; the weighted sum $\widetilde Z=\sum_tc_t\widetilde\omega_t$, with nonnegative integer weights totaling at most $K$, then has $O(P+\log K)$ numerator bits. Squared-magnitude comparisons use integers of polynomial bit length. \Cref{sec:paper-resources} constructs these approximations and bounds the cost of list construction and processing the root entries.

For the bounded-secret sample range of the main theorem, $m\ge Cn/\ln\ln q$ implies $d=\min\{m,n\}\in[Cn/\ln\ln q,n]$ for sufficiently large $n$; for the arbitrary-secret range, $m\ge n+Cn/\ln\ln q$ similarly gives $m>n$ and $d=\min\{m-n,n\}$ in the same interval. Thus the parameter checks succeed for sufficiently large inputs throughout both ranges. The retained positions depend only on the integer dimensions $m,n$, and the parameter choice and its checks depend only on $n,q,d,E_s,E_e$. The sample-size inequalities involving $C$ and $\ln\ln q$ are used only in the analysis.
\end{proof}

This verifies the conditions assumed in Sections~\ref{sec:paper-moments} and~\ref{sec:paper-resources}, and supplies the signal and running-time exponent used to assemble the main theorem in Section~\ref{sec:paper-reductions}.

\begin{samepage}
\section{AI Acknowledgments}\

AI was used in the production of this work. Building on the second-moment method developed in our earlier work on $k$-trees over integers~\cite{LV24}, we adapted this technique to the current setting, found ways to deal with the greater dependencies here, and essentially reduced the analysis to bounding the covariances among a set of indicator variables. OpenAI's ChatGPT assisted with refining the required moment and covariance bounds and the parameter choices. This involved sustained interaction and verification over several weeks and passing over multiple earlier proof approaches that were needlessly complex, until we arrived at the current proof.

ChatGPT was also used for writing throughout the paper, including preparing all the figures and tables. Anthropic's Claude was used for verifying correctness, in addition to our own verification. The authors take responsibility for all content in this paper.
\end{samepage}

% \clearpage
\phantomsection
\label{sec:references-start}
\bibliographystyle{alphaurl}
\bibliography{references}

@string{springer =              "Springer"}

@string{springer_berlin_heidelberg = "Springer Berlin Heidelberg, Germany"}

@string{springer_cham =         "Springer, Cham, Switzerland"}

@string{jcrypto =               "Journal of Cryptology"}

@string{dcc =                   "Designs, Codes and Cryptography"}

@string{mylncs =                "Lecture Notes in Computer Science"}

@string{asiacrypt25name1 =      asiacryptname # "~2025, Part~I"}

@string{asiacrypt25ed =         "Goichiro Hanaoka and Bo-Yin Yang"}

@string{asiacrypt25vol1 =       "16245"}

@string{asiacrypt25addr =       "Melbourne, VIC, Australia"}

@string{asiacrypt25month =      dec # "~8--12,"}

@string{asiacrypt25pub =        springer_singapore}

@string{cryptoaddr =            "Santa Barbara, CA, USA"}

@string{crypto02name =          cryptoname # "~2002"}

@string{crypto02ed =            "Moti Yung"}

@string{crypto02vol =           "2442"}

@string{crypto02month =         aug # "~18--22,"}

@string{crypto02pub =           springer_berlin_heidelberg}

@string{crypto13name1 =         cryptoname # "~2013, Part~I"}

@string{crypto13ed =            "Ran Canetti and Juan A. Garay"}

@string{crypto13vol1 =          "8042"}

@string{crypto13month =         aug # "~18--22,"}

@string{crypto13pub =           springer_berlin_heidelberg}

@string{crypto15name1 =         cryptoname # "~2015, Part~I"}

@string{crypto15ed =            "Rosario Gennaro and Matthew J. B. Robshaw"}

@string{crypto15vol1 =          "9215"}

@string{crypto15month =         aug # "~16--20,"}

@string{crypto15pub =           springer_berlin_heidelberg}

@string{crypto25name1 =         cryptoname # "~2025, Part~I"}

@string{crypto25ed =            "Yael Tauman Kalai and Seny F. Kamara"}

@string{crypto25vol1 =          "16000"}

@string{crypto25month =         aug # "~17--21,"}

@string{crypto25pub =           springer_cham}

@string{icalppub =              springer_berlin_heidelberg}

@string{icalp11name1 =          "ICALP 2011: 38th " # icalpname # ", Part~I"}

@string{icalp11ed =             "Luca Aceto and Monika Henzinger and Jiri Sgall"}

@string{icalp11vol1 =           "6755"}

@string{icalp11addr =           "Zurich, Switzerland"}

@string{icalp11month =          jul # "~4--8,"}

@string{indocrypt20name =       indocryptname # "~2020: 21st " # indocryptsuffix}

@string{indocrypt20ed =         "Karthikeyan Bhargavan and Elisabeth Oswald and Manoj Prabhakaran"}

@string{indocrypt20vol =        "12578"}

@string{indocrypt20pub =        springer_cham}

@string{stocpub =               "{ACM} Press"}

@string{stoc96name =            "28th " # stocname}

@string{stoc96addr =            "Philadelphia, PA, USA"}

@string{stoc96month =           may # "~22--24,"}

@string{stoc05name =            "37th " # stocname}

@string{stoc05ed =              "Harold N. Gabow and Ronald Fagin"}

@string{stoc05addr =            "Baltimore, MD, USA"}

@string{stoc05month =           may # "~22--24,"}

@Misc{HGH26,
  author =       {Xuyuan Han and Yiming Gao and Honggang Hu},
  title =        {A Provable Subexponential-Time Algorithm for {LWE} from {$n+o(n)$} Samples via {W}agner-Style Gaussian Sampling},
  year =         2026,
  howpublished = {Cryptology {ePrint} Archive, Paper 2026/2261},
  url =          {https://eprint.iacr.org/2026/2261},
  note =         {Uploaded 29 September 2026},
}

@Misc{CJWZ26,
  author =       {Zhili Chen and Kaijie Jiang and Anyu Wang and Chuanqi Zhang},
  title =        {A Subexponential Algorithm for Binary {LWE} with Sublinear Samples},
  year =         2026,
  howpublished = {Cryptology {ePrint} Archive, Paper 2026/2289},
  url =          {https://eprint.iacr.org/2026/2289},
  note =         {Uploaded 1 October 2026},
}

@InProceedings{LV24,
  author =       "Haoxing Lin and
                  Prashant Nalini Vasudevan",
  title =        {On {W}agner's {$k$}-Tree Algorithm Over Integers},
  pages =        {481--513},
  editor =       asiacrypt25ed,
  booktitle =    asiacrypt25name1,
  volume =       asiacrypt25vol1,
  address =      asiacrypt25addr,
  month =        asiacrypt25month,
  publisher =    asiacrypt25pub,
  series =       mylncs,
  year =         2025,
  doi =          "10.1007/978-981-95-5018-0_16",
  note = {Full version: \url{https://eprint.iacr.org/2024/1612}},
}

@InProceedings{DEL25,
  author =       "L{\'e}o Ducas and
                  Lynn Engelberts and
                  Johanna Loyer",
  title =        {Wagner's Algorithm Provably Runs in Subexponential Time
               for {$\mathrm{SIS}^{\infty}$}},
  pages =        {353--384},
  editor =       crypto25ed,
  booktitle =    crypto25name1,
  volume =       crypto25vol1,
  address =      cryptoaddr,
  month =        crypto25month,
  publisher =    crypto25pub,
  series =       mylncs,
  year =         2025,
  doi =          "10.1007/978-3-032-01855-7_12",
  note = {Full version: \url{https://eprint.iacr.org/2025/575}; correction: \url{https://doi.org/10.1007/978-3-032-01855-7_23}, p.~C1},
}

@InProceedings{Wag02,
  author =       "David Wagner",
  title =        {A Generalized Birthday Problem},
  pages =        {288--304},
  editor =       crypto02ed,
  booktitle =    crypto02name,
  volume =       crypto02vol,
  address =      cryptoaddr,
  month =        crypto02month,
  publisher =    crypto02pub,
  series =       mylncs,
  year =         2002,
  doi =          "10.1007/3-540-45708-9_19",
}

@InProceedings{KF15,
  author =       "Paul Kirchner and
                  Pierre-Alain Fouque",
  title =        {An Improved {BKW} Algorithm for {LWE} with Applications
               to Cryptography and Lattices},
  pages =        {43--62},
  editor =       crypto15ed,
  booktitle =    crypto15name1,
  volume =       crypto15vol1,
  address =      cryptoaddr,
  month =        crypto15month,
  publisher =    crypto15pub,
  series =       mylncs,
  year =         2015,
  doi =          "10.1007/978-3-662-47989-6_3",
}

@Article{HKM18,
  author =       "Gottfried Herold and
                  Elena Kirshanova and
                  Alexander May",
  title =        {On the Asymptotic Complexity of Solving {LWE}},
  pages =        {55--83},
  volume =       86,
  publisher =    springer,
  year =         2018,
  journal =      dcc,
  number =       1,
  doi =          "10.1007/s10623-016-0326-0",
}

@Misc{JKL24,
  author =       "Antoine Joux and
                  Hunter Kippen and
                  Julian Loss",
  title =        {A Concrete Analysis of {W}agner's {$k$}-List Algorithm over {$\mathbb{Z}_p$}},
  year =         2024,
  howpublished = "Cryptology {ePrint} Archive, Report 2024/282",
  url =          "https://eprint.iacr.org/2024/282",
}

@InProceedings{Ajtai96,
  author =       "Mikl{\'o}s Ajtai",
  title =        {Generating Hard Instances of Lattice Problems},
  pages =        {99--108},
  booktitle =    stoc96name,
  address =      stoc96addr,
  month =        stoc96month,
  publisher =    stocpub,
  year =         1996,
  doi =          "10.1145/237814.237838",
}

@Article{MinderSinclair12,
  author =       "Lorenz Minder and
                  Alistair Sinclair",
  title =        {The Extended {$k$}-Tree Algorithm},
  pages =        {349--382},
  volume =       25,
  month =        apr,
  publisher =    springer,
  year =         2012,
  journal =      jcrypto,
  number =       2,
  doi =          "10.1007/s00145-011-9097-y",
}

@InProceedings{MP13,
  author =       "Daniele Micciancio and
                  Chris Peikert",
  title =        {Hardness of {SIS} and {LWE} with Small Parameters},
  pages =        {21--39},
  editor =       crypto13ed,
  booktitle =    crypto13name1,
  volume =       crypto13vol1,
  address =      cryptoaddr,
  month =        crypto13month,
  publisher =    crypto13pub,
  series =       mylncs,
  year =         2013,
  doi =          "10.1007/978-3-642-40041-4_2",
}

@InProceedings{AG11,
  author =       "Sanjeev Arora and
                  Rong Ge",
  title =        {New Algorithms for Learning in Presence of Errors},
  pages =        {403--415},
  editor =       icalp11ed,
  booktitle =    icalp11name1,
  volume =       icalp11vol1,
  address =      icalp11addr,
  month =        icalp11month,
  publisher =    icalppub,
  series =       mylncs,
  year =         2011,
  doi =          "10.1007/978-3-642-22006-7_34",
}

@InProceedings{Reg05,
  author =       "Oded Regev",
  title =        {On Lattices, Learning with Errors, Random Linear Codes, and Cryptography},
  pages =        {84--93},
  editor =       stoc05ed,
  booktitle =    stoc05name,
  address =      stoc05addr,
  month =        stoc05month,
  publisher =    stocpub,
  year =         2005,
  doi =          "10.1145/1060590.1060603",
}

@article{BKW03,
  author    = {Avrim Blum and Adam Kalai and Hal Wasserman},
  title     = {Noise-Tolerant Learning, the Parity Problem, and the
               Statistical Query Model},
  journal   = {Journal of the {ACM}},
  volume    = {50},
  number    = {4},
  pages     = {506--519},
  year      = {2003},
  doi       = {10.1145/792538.792543},
}

@book{BrentZimmermann10,
  author    = {Richard P. Brent and Paul Zimmermann},
  title     = {Modern Computer Arithmetic},
  publisher = {Cambridge University Press},
  year      = {2010},
  url       = {https://members.loria.fr/PZimmermann/mca/mca-cup-0.5.9.pdf},
}

@article{AR05,
  author  = {Dorit Aharonov and Oded Regev},
  title   = {Lattice Problems in {NP} {$\cap$} co{NP}},
  journal = {Journal of the {ACM}},
  volume  = {52},
  number  = {5},
  pages   = {749--765},
  year    = {2005},
  doi     = {10.1145/1089023.1089025},
}

@inproceedings{GMSW21,
  author    = {Qian Guo and Erik M{\aa}rtensson and Stankovski Wagner, Paul},
  title     = {On the Sample Complexity of Solving {LWE} Using {BKW}-Style Algorithms},
  booktitle = {2021 {IEEE} International Symposium on Information Theory ({ISIT})},
  pages     = {2405--2410},
  publisher = {{IEEE}},
  year      = {2021},
  doi       = {10.1109/ISIT45174.2021.9518190},
}

@inproceedings{ACPS09,
  author    = {Benny Applebaum and David Cash and Chris Peikert and Amit Sahai},
  title     = {Fast Cryptographic Primitives and Circular-Secure Encryption Based on Hard Learning Problems},
  booktitle = {Advances in Cryptology -- {CRYPTO} 2009},
  series    = {Lecture Notes in Computer Science},
  volume    = {5677},
  pages     = {595--618},
  publisher = {Springer},
  year      = {2009},
  doi       = {10.1007/978-3-642-03356-8_35},
}

@inproceedings{Lyu05,
  author    = {Vadim Lyubashevsky},
  title     = {The Parity Problem in the Presence of Noise, Decoding Random Linear Codes, and the Subset Sum Problem},
  booktitle = {Approximation, Randomization and Combinatorial Optimization. Algorithms and Techniques ({APPROX}/{RANDOM} 2005)},
  series    = {Lecture Notes in Computer Science},
  volume    = {3624},
  pages     = {378--389},
  publisher = {Springer},
  year      = {2005},
  doi       = {10.1007/11538462_32},
}

@inproceedings{ACKS21,
  author    = {Divesh Aggarwal and Yanlin Chen and Rajendra Kumar and Yixin Shen},
  title     = {Improved (Provable) Algorithms for the Shortest Vector Problem via Bounded Distance Decoding},
  editor    = {Markus Bl{\"a}ser and Benjamin Monmege},
  booktitle = {38th International Symposium on Theoretical Aspects of Computer Science ({STACS} 2021)},
  series    = {Leibniz International Proceedings in Informatics (LIPIcs)},
  volume    = {187},
  pages     = {4:1--4:20},
  publisher = {Schloss Dagstuhl -- Leibniz-Zentrum f{\"u}r Informatik},
  address   = {Dagstuhl, Germany},
  year      = {2021},
  doi       = {10.4230/LIPIcs.STACS.2021.4},
}

@inproceedings{ADRS15,
  author    = {Divesh Aggarwal and Daniel Dadush and Oded Regev and Noah Stephens-Davidowitz},
  title     = {Solving the Shortest Vector Problem in {$2^n$} Time Using Discrete {G}aussian Sampling: Extended Abstract},
  booktitle = {Proceedings of the Forty-Seventh Annual {ACM} Symposium on Theory of Computing ({STOC} 2015)},
  pages     = {733--742},
  publisher = {{ACM}},
  year      = {2015},
  doi       = {10.1145/2746539.2746606},
  note      = {Full version: \url{https://arxiv.org/abs/1412.7994}},
}

@inproceedings{ADS15,
  author    = {Divesh Aggarwal and Daniel Dadush and Noah Stephens-Davidowitz},
  title     = {Solving the Closest Vector Problem in {$2^n$} Time -- The Discrete {G}aussian Strikes Again!},
  booktitle = {2015 {IEEE} 56th Annual Symposium on Foundations of Computer Science ({FOCS})},
  pages     = {563--582},
  publisher = {{IEEE} Computer Society},
  year      = {2015},
  doi       = {10.1109/FOCS.2015.41},
  note      = {Full version: \url{https://arxiv.org/abs/1504.01995}},
}

@inproceedings{AFFP14,
  author    = {Martin R. Albrecht and Jean-Charles Faug{\`e}re and Robert Fitzpatrick and Ludovic Perret},
  title     = {Lazy Modulus Switching for the {BKW} Algorithm on {LWE}},
  editor    = {Hugo Krawczyk},
  booktitle = {Public-Key Cryptography -- {PKC} 2014},
  series    = mylncs,
  volume    = {8383},
  pages     = {429--445},
  publisher = springer_berlin_heidelberg,
  year      = {2014},
  doi       = {10.1007/978-3-642-54631-0_25},
}

@inproceedings{ALS21,
  author    = {Divesh Aggarwal and Zeyong Li and Noah Stephens-Davidowitz},
  title     = {A {$2^{n/2}$}-Time Algorithm for {$\sqrt{n}$}-{SVP} and {$\sqrt{n}$}-{Hermite SVP}, and an Improved Time-Approximation Tradeoff for {(H)SVP}},
  editor    = {Anne Canteaut and Fran{\c c}ois-Xavier Standaert},
  booktitle = {Advances in Cryptology -- {EUROCRYPT} 2021, Part~I},
  series    = mylncs,
  volume    = {12696},
  pages     = {467--497},
  publisher = springer_cham,
  year      = {2021},
  doi       = {10.1007/978-3-030-77870-5_17},
}

@inproceedings{AS18,
  author    = {Divesh Aggarwal and Noah Stephens-Davidowitz},
  title     = {Just Take the Average! An Embarrassingly Simple {$2^n$}-Time Algorithm for {SVP} (and {CVP})},
  editor    = {Raimund Seidel},
  booktitle = {1st Symposium on Simplicity in Algorithms ({SOSA} 2018)},
  series    = {Open Access Series in Informatics (OASIcs)},
  volume    = {61},
  pages     = {12:1--12:19},
  publisher = {Schloss Dagstuhl -- Leibniz-Zentrum f{\"u}r Informatik},
  address   = {Dagstuhl, Germany},
  year      = {2018},
  doi       = {10.4230/OASIcs.SOSA.2018.12},
}

@inproceedings{BGJMSW20,
  author    = {Alessandro Budroni and Qian Guo and Thomas Johansson and Erik M{\aa}rtensson and Stankovski Wagner, Paul},
  title     = {Making the {BKW} Algorithm Practical for {LWE}},
  editor    = indocrypt20ed,
  booktitle = indocrypt20name,
  series    = mylncs,
  volume    = indocrypt20vol,
  pages     = {417--439},
  publisher = indocrypt20pub,
  year      = {2020},
  doi       = {10.1007/978-3-030-65277-7_19},
}

@inproceedings{BLP13,
  author    = {Zvika Brakerski and Adeline Langlois and Chris Peikert and Oded Regev and Damien Stehl{\'e}},
  title     = {Classical Hardness of Learning with Errors},
  booktitle = {Proceedings of the Forty-Fifth Annual {ACM} Symposium on Theory of Computing ({STOC} 2013)},
  pages     = {575--584},
  publisher = {{ACM}},
  year      = {2013},
  doi       = {10.1145/2488608.2488680},
  note      = {Full version: \url{https://arxiv.org/abs/1306.0281}},
}

@inproceedings{GJMS17,
  author    = {Qian Guo and Thomas Johansson and Erik M{\aa}rtensson and Paul Stankovski},
  title     = {Coded-{BKW} with Sieving},
  editor    = {Tsuyoshi Takagi and Thomas Peyrin},
  booktitle = {Advances in Cryptology -- {ASIACRYPT} 2017, Part~I},
  series    = mylncs,
  volume    = {10624},
  pages     = {323--346},
  publisher = springer_cham,
  year      = {2017},
  doi       = {10.1007/978-3-319-70694-8_12},
}

@inproceedings{LF06,
  author    = {{\'E}ric Levieil and Pierre-Alain Fouque},
  title     = {An Improved {LPN} Algorithm},
  editor    = {Roberto De Prisco and Moti Yung},
  booktitle = {Security and Cryptography for Networks ({SCN} 2006)},
  series    = mylncs,
  volume    = {4116},
  pages     = {348--359},
  publisher = springer_berlin_heidelberg,
  year      = {2006},
  doi       = {10.1007/11832072_24},
}

@article{MR07,
  author    = {Daniele Micciancio and Oded Regev},
  title     = {Worst-Case to Average-Case Reductions Based on {G}aussian Measures},
  journal   = {{SIAM} Journal on Computing},
  volume    = {37},
  number    = {1},
  pages     = {267--302},
  publisher = {Society for Industrial and Applied Mathematics},
  year      = {2007},
  doi       = {10.1137/S0097539705447360},
}

@inproceedings{WBLW25,
  author    = {Yu Wei and Lei Bi and Xianhui Lu and Kunpeng Wang},
  title     = {Memory-Efficient {BKW} Algorithm for Solving the {LWE} Problem},
  editor    = {Tibor Jager and Jiaxin Pan},
  booktitle = {Public-Key Cryptography -- {PKC} 2025, Part~II},
  series    = mylncs,
  volume    = {15675},
  pages     = {331--362},
  publisher = springer_cham,
  year      = {2025},
  doi       = {10.1007/978-3-031-91823-0_11},
  note      = {Full version: \url{https://eprint.iacr.org/2025/266}},
}

% \clearpage
\appendix
\section{Standard reductions with a fixed number of original samples}
\label{app:reduction-proofs}

We verify the forms of the standard reductions used in \Cref{sec:paper-reductions}. The points requiring care are the fixed secret-error pair, reuse of the original samples, and the explicit failure bounds.

\subsection{Bounded-secret search on the same samples}
\label{app:bounded-search}

Regev's search-to-decision reduction~\cite[\href{https://cims.nyu.edu/~regev/papers/qcrypto.pdf}{full version}, Lemma~4.2]{Reg05} is stated for independent samples with errors drawn from one distribution and a decision procedure with exponentially small error for every field secret. Here the decision guarantee is $3n^{-4}$ for bounded secrets and every fixed error vector, and all guesses must use the same original samples. The same coordinate perturbation applies; we record its marginal distributions and the resulting union bound.

\begin{proof}[Proof of \Cref{prop:bounded-search}]
Fix a secret-error pair satisfying the stated bounds. Write $A,b$ for the matrix and right-hand-side vector of the first $d$ samples, and $e$ for their error vector. For each coordinate $j=1,\ldots,n$ in turn, run a decision call for every guess $\sigma\in\mathcal S$ in increasing order, with $\mathcal S$ as in \Cref{prop:bounded-search}. For a call $(j,\sigma)$, let $u_j$ be the $j$th standard basis vector in $\Fq^n$ and draw a fresh uniform $w\in\Fq^d$. The perturbed input is
\begin{equation}
\label{eq:reduction-perturbation}
 \widehat A=A+u_jw^\top,\qquad
 \widehat b=b+\sigma w=\widehat A^\top s+e+(\sigma-s_j)w.
\end{equation}
For every fixed $w$, each possible $\widehat A$ is produced by exactly one original matrix, $A=\widehat A-u_jw^\top$. All these original matrices are equally likely, so $\widehat A$ is uniform with the same distribution for every $w$. Thus $\widehat A$ and $w$ are independent uniform variables.

If $\sigma=s_j\bmod q$, \eqref{eq:reduction-perturbation} gives $\widehat b=\widehat A^\top s+e$, so the input has the same fixed pair $(s,e)$. Otherwise, $\sigma-s_j$ is nonzero modulo the prime $q$, so multiplying by it permutes the vectors in $\Fq^d$. Conditional on $\widehat A$, $w$ is still uniform, so $(\sigma-s_j)w$ is uniform. Adding the fixed vector $\widehat A^\top s+e$ preserves this uniformity, giving a uniform $\widehat b$ conditional on $\widehat A$. Each call receives only the perturbed input and the public bounds, with independent fresh tree randomness.

If any parameter check in Section~\ref{sec:parameter-computation} fails, report failure before the calls. For sufficiently large $n$, $Cn/\ln\ln q\le d\le n$, so these checks pass by \Cref{lem:schedule-properties}. Also, $E_s\le n^{o(1)}$, so $2\lfloor E_s\rfloor<q$ for sufficiently large $n$. The guesses are therefore distinct modulo $q$ and include each true coordinate exactly once, including when $E_s=0$.

By \Cref{prop:bounded-input-decision}, each call $(j,\sigma)$ has answer error probability at most $3n^{-4}$, averaged over the original matrix, its perturbation vector, and its tree randomness. If every answer is correct, search recovers every coordinate. A union bound over all $M_{\mathrm{calls}}$ calls therefore gives failure probability at most $3M_{\mathrm{calls}}n^{-4}=n^{-3+o(1)}$. This argument neither conditions on earlier answers nor requires independence between calls.
\end{proof}

\subsection{Normal form from the reserved samples}
\label{app:fixed-suffix-normal-form}

The normal-form lemma of Applebaum--Cash--Peikert--Sahai~\cite[Section~3.1, Lemma~2]{ACPS09} is stated for independent LWE samples. Its preprocessing collects linearly independent columns, discarding samples whose columns are linearly dependent on those already collected, before transforming fresh samples. We instead reserve exactly $n$ fixed sample positions and allow failure when that block is singular. We also verify the transformation for every fixed error vector, so the guarantees hold even for error distributions with correlated coordinates.

\begin{proof}[Proof of \Cref{prop:fixed-suffix-normal-form}]
Let $A_{\mathrm{pre}}\in\Fq^{n\times(m-n)}$ be the matrix of the first $m-n$ equations and $A_B\in\Fq^{n\times n}$ the matrix of the final $n$ equations. Write $b_{\mathrm{pre}},b_B$ and $e_{\mathrm{pre}},e_B$ for their right-hand-side and error vectors. Test $A_B$ for invertibility once and report failure if it is singular. Otherwise output
\begin{equation}
\label{eq:reduction-normal-form}
 \begin{gathered}
 A^{\mathrm{nf}}=A_B^{-1}A_{\mathrm{pre}},\qquad
 b^{\mathrm{nf}}=b_{\mathrm{pre}}-(A^{\mathrm{nf}})^\top b_B,\\
 s^{\mathrm{nf}}=-e_B,\qquad e^{\mathrm{nf}}=e_{\mathrm{pre}}.
 \end{gathered}
\end{equation}
Only $A^{\mathrm{nf}}$ and $b^{\mathrm{nf}}$ are computed; the transformed secret and errors specify their LWE representation.

Fix an original pair $(s,e)$ and an invertible matrix of reserved samples $A_B$. Multiplication by $A_B^{-1}$ carries the independent uniform matrix $A_{\mathrm{pre}}$ bijectively to uniform $A^{\mathrm{nf}}$. Substituting $b_{\mathrm{pre}}=A_{\mathrm{pre}}^\top s+e_{\mathrm{pre}}$ and $b_B=A_B^\top s+e_B$ gives
\begin{equation}
\label{eq:reduction-normal-form-law}
 b^{\mathrm{nf}}=e_{\mathrm{pre}}-(A^{\mathrm{nf}})^\top e_B
   =(A^{\mathrm{nf}})^\top s^{\mathrm{nf}}+e^{\mathrm{nf}}.
\end{equation}
The transformed secret-error pair $(-e_B,e_{\mathrm{pre}})$ is fixed, has dimensions $n$ and $m-n$, and has both coordinate bounds at most $E$. In the uniform case, condition further on $A^{\mathrm{nf}},b_B$: the right-hand-side vector $b_{\mathrm{pre}}$ remains independent uniform, and hence so does $b^{\mathrm{nf}}$. Since $b_B+s^{\mathrm{nf}}=A_B^\top s$, a recovered transformed secret gives
\begin{equation}
\label{eq:reduction-recovery}
 s=A_B^{-\top}(b_B+s^{\mathrm{nf}})\pmod q.
\end{equation}
The invertibility test, transformation, and recovery use Gaussian elimination and field arithmetic, taking time polynomial in the input length.

To bound the only transformation failure, reveal the columns of $A_B$ in order. Conditional on the first $j$ columns being linearly independent, the next lies in their span with probability $q^{j-n}$. Thus
\begin{equation}
\label{eq:reduction-suffix-singularity}
 \Prb[A_B\text{ singular}]
 \le\sum_{j=0}^{n-1}q^{j-n}=\frac{1-q^{-n}}{q-1}\le\frac1{q-1}.
\end{equation}
\end{proof}

\begin{proof}[Proof of \Cref{cor:arbitrary-secret}]
For $m\le n$, define arbitrary-secret decision to output \textsc{uniform} and search to report failure. Otherwise test the reserved matrix $A_B$ once; on singularity, decision returns \textsc{uniform} and search reports failure. On invertibility, apply \Cref{prop:fixed-suffix-normal-form} and use the first $d=\min\{m-n,n\}$ transformed equations for bounded-secret decision or search with $E_s=E_e=E$. Computing only these $d$ columns in \eqref{eq:reduction-normal-form} is equivalent to transforming the whole prefix and then discarding the unused equations. Search perturbs only these transformed equations, keeping the reserved samples and $b_B$ unchanged for recovery by \eqref{eq:reduction-recovery}.

For sufficiently large $n$, the sample assumption gives $m>n$ and $Cn/\ln\ln q\le d\le n$. The bound $E\le\exp((\ln n)^{o(1)})$ also implies $E\le n^{o(1)}$. Conditional on each fixed invertible reserved matrix, \Cref{prop:bounded-input-decision,prop:bounded-search} therefore apply to the fixed transformed secret-error pair. Average their conditional bounds and add \eqref{eq:reduction-suffix-singularity} once for the entire computation to obtain \eqref{eq:reduction-arbitrary-errors}. Uniform input remains uniform, and a singular reserved matrix yields the correct decision output \textsc{uniform}, so the stated decision bound also holds in that case. Both bounds tend to zero for every $\kappa>0$, including $q<n$ and $E=0$.
\end{proof}

\clearpage
\section{Concurrent work}
\label{app:concurrent-work}

% Independent concurrent works of Han, Gao, and Hu~\cite{HGH26} and Chen, Jiang, Wang, and Zhang~\cite{CJWZ26} were uploaded to the IACR ePrint Archive on 29 September and 1 October 2026, respectively, while we were completing the final revision of this manuscript. An earlier version containing our results and their proofs had already been completed by 11 September 2026, so we consider our results as independent and concurrent work.\footnote{Our repository records such a version in commit \href{https://github.com/haoxingl/Provable-Wagners-for-SIS/commit/343f965f0316ac16fbebac302dfb4ab800c8e1a4}{\texttt{343f965}}, dated 11 September 2026.}

Independent concurrent works by Han, Gao, and Hu~\cite{HGH26} and Chen, Jiang, Wang, and Zhang~\cite{CJWZ26} contain results very similar to ours, and we discuss these similarities and differences here.

An important different between our work and these is the approach we use in our analysis, which we consider our most important contribution. The concurrent works use a different analytical framework from ours. Both develop the Gaussian--Wagner framework of Ducas, Engelberts, and Loyer~\cite{DEL25}, introducing new auxiliary superlattices that support the Gaussian sampling accuracy needed for LWE. Our approach extends the second-moment analysis of Lin and Vasudevan~\cite{LV24}: we control the Fourier statistic directly for a reusable $k$-tree, including the dependence created by the shared public matrix. 

Han--Gao--Hu address arbitrary secrets with $n+o(n)$ samples, whereas Chen--Jiang--Wang--Zhang address binary and, more generally, narrow secrets and errors with sublinear samples. Taken together, they provide Gaussian--Wagner analogues of our two main sample regimes, with several differences in scope and guarantees.

One difference is the modulus range. For polynomial prime moduli $q=n^{\kappa+o(1)}$, both papers state their main specializations for fixed $\kappa>1/2$. Our results allow every fixed $\kappa>0$, including moduli that grow more slowly than $\sqrt n$.

For arbitrary secrets, Han--Gao--Hu obtain a wider error range at the cost of a larger running time and more samples. Their theorem covers independent discrete Gaussian errors whose width can grow as $\exp((\ln n)^{a+o(1)})$ for any fixed $0<a<1$, beyond the range stated in our theorem, with a corresponding increase in the running-time exponent. Our theorem is formulated for every fixed bounded secret-error pair, allowing correlated errors; their headline theorem uses independent Gaussian errors.

For narrow secrets, Chen--Jiang--Wang--Zhang state a fixed bound on the average squared size of the secret and errors, whereas our theorem permits secret coordinates bounded by $\exp((\ln n)^{o(1)})$ and error coordinates bounded by $n^{o(1)}$. The two promises cover different inputs: a norm bound permits a few much larger coordinates, while our bounds permit all coordinates to grow. With $\Theta(n/\ln\ln n)$ samples, their general sampler also appears to accommodate the polylogarithmic special case after a change of parameters, so that difference in stated scope should not be viewed as a limitation of their framework. A more substantial distinction is running time: at this sample scale, their proved bound has leading constant $8$ in the base-two exponent, compared with $1/2$ here. They also give a broader sample--time tradeoff for their stated norm promise, obtaining subexponential time with any $m=\omega(n/\ln n)$, $m=o(n)$, at the cost of a larger running time as the sample count decreases.

% Finally, both papers' proved LWE results concern decision. Han--Gao--Hu discuss search under an assumed reduction that preserves the number of samples; our theorem proves both decision and search with the stated number of original samples. Their resource bounds are in expectation, whereas ours gives a bound on bit time with high probability and on bit space in the worst case.

\end{document}